\documentclass[journal]{IEEEtran}
\usepackage{comment}
\usepackage{epsfig,graphics,amssymb}
\usepackage{amscd}
\usepackage{amsmath}
\usepackage[normal,bf]{caption2}
\usepackage{enumerate}
\usepackage{theorem}
\usepackage{cite}
\usepackage{epsfig}
\usepackage{verbatim}
\usepackage{color}
\theoremstyle{plain} \theorembodyfont{\itshape}
\theoremheaderfont{\scshape}
\newtheorem{theorem}{Theorem}
\theorembodyfont{\itshape} \theoremheaderfont{\scshape}

\newtheorem{lemma}{Lemma}

\theoremstyle{plain} \theorembodyfont{\itshape}
\theoremheaderfont{\scshape}
\newtheorem{corollary}{Corollary}
\newtheorem{proposition}{Proposition}

\newtheorem{example}{Example}

\newtheorem{assumption}{Assumption}
\newtheorem{proof}{Proof}
\newtheorem{definition}{Definition}
\newtheorem{remark}{Remark}

\newcommand{\new}{\color{black}}
\newtheorem{notation}{Notation}
\begin{document}
\title{Capacity Scaling in MIMO Systems with General Unitarily Invariant Random Matrices}
\author{Burak~\c{C}akmak, Ralf~R.~M\"{u}ller,~\IEEEmembership{Senior~Member,~IEEE,} Bernard~H.~Fleury,~\IEEEmembership{Senior~Member,~IEEE}% <-this % stops a space
\thanks{%Manuscript received December 08, 2015; revised July 11, 2017; accepted February 4, 2018. Date of publication  XXXXXX, 2018; date of current version XXXXXX, 2018. 
Burak~\c{C}akmak and Bernard~H.~Fleury were supported by the research project VIRTUOSO funded by Intel Mobile Communications, Keysight, Telenor, Aalborg University, and the Danish National Advanced Technology Foundation. Ralf~R.~M\"{u}ller was supported by the Alexander von Humboldt Foundation.}
\thanks{Burak~\c{C}akmak is  with the Department of Computer Science, Technical University of Berlin, 10587 Berlin, Germany (e-mail: burak.cakmak@tu-berlin.de).} 
\thanks{Ralf~R.~M\"{u}ller is with the Institute for Digital Communications, Friedrich-Alexander Universit{\"a}t Erlangen-N\"{u}rnberg,  91058 Erlangen, Germany (e-mail: mueller@lnt.de).}
\thanks{Bernard~H.~Fleury is with the Department of Electronic Systems, Aalborg University, 9220 Aalborg, Denmark (e-mail: fleury@es.aau.dk).}
%\thanks{Communicated by M. K. Varanasi, Associate Editor for Communications. Color versions of one or more of the figures in this paper are available online at http://ieeexplore.ieee.org.}
%\thanks{Digital Object Identifier XXXXXX}
}
%\IEEEpubid{\begin{minipage}{\textwidth}\ \\[12pt] \centering 
%	\copyright 2017 IEEE. Personal use is permitted. However, permission to use this material for any other \\ purposes must be obtained from the IEEE by sending a request to pubs-permissions@ieee.org.
%\end{minipage}}
\markboth{IEEE Transactions on Information Theory, 2018}{\c{C}akmak \MakeLowercase{\text{et al.}}: Capacity Scaling in MIMO Systems with General Unitarily Invariant Random Matrices} 
\maketitle
\begin{abstract}
We investigate the capacity scaling of MIMO systems with the system dimensions. To that end we quantify how the mutual information varies when the number of antennas (at either the receiver or transmitter side) is altered. For a system comprising $R$ receive and $T$ transmit antennas with $R>T$, we find the following: By removing as many receive antennas as needed to obtain a square system (provided the channel matrices before and after the removal have full rank) the maximum resulting loss of mutual information over all signal-to-noise ratios (SNRs) depends only on $R$, $T$ and the matrix of left-singular vectors of the initial channel matrix, but not on its singular values. In particular, if the latter matrix is Haar distributed the ergodic rate loss is given by $\sum_{t=1}^{T}\sum_{r=T+1}^{R}\frac{1}{r-t}$ nats. Under the same assumption, if $T,R\to \infty$ with the ratio $\phi\triangleq T/R$ fixed, the rate loss normalized by $R$ converges almost surely to $H(\phi)$~bits with $H(\cdot)$ denoting the binary entropy function. We also quantify and study how the mutual information as a function of the system dimensions deviates from the traditionally assumed linear growth in
the minimum of the system dimensions at high SNR. 
\end{abstract}

\def\mathlette#1#2{{\mathchoice{\mbox{#1$\displaystyle #2$}}%
                               {\mbox{#1$\textstyle #2$}}%
                               {\mbox{#1$\scriptstyle #2$}}%
                               {\mbox{#1$\scriptscriptstyle #2$}}}}
\newcommand{\matr}[1]{\mathlette{\boldmath}{#1}}
\newcommand{\RR}{\mathbb{R}}
\newcommand{\CC}{\mathbb{C}}
\newcommand{\NN}{\mathbb{N}}
\newcommand{\ZZ}{\mathbb{Z}}
\newcommand{\EE}{\mathbb{E}}
\section{Introduction}
\IEEEPARstart{T}{he} capacity of a multiple-input--multiple-output (MIMO) system with perfect channel state information at the receiver can be expressed as \cite{shamai}
\begin{equation}
\min (T,R)\log_2 {\rm SNR}+O(1) \label{approx}
\end{equation}
whenever the channel matrix has full rank almost surely. Here $T$ and $R$ denote the number of receive and transmit antennas, respectively, and $O(1)$ is a bounded function of the signal-to-noise ratio (SNR) \emph{that does depend on $T$ and $R$, in general.} The scaling term $\min(T,R)$ is often referred to as the multiplexing gain. The explicit expression for the capacity scaling when the number of transmit or receive antennas varies, is difficult to calculate. Closed-form expressions can be obtained only in few particular cases, e.g.\ for a channel matrix of asymptotically large size with independent identically distributed (iid) zero-mean entries \cite{verdu}. 

In order to better understand capacity scaling in MIMO channels with more complicated structures, such as correlation at transmit and/or receive antennas, related works use either implicit solutions, e.g.\ \cite{lozano}, or consider asymptotically high SNR and express the capacity in terms of the multiplexing gain, e.g.\ \cite{mimo3}. However, implicit solutions provide limited intuitive insight into the capacity scaling and the multiplexing gain is a crude measure of capacity. 

In this article, we consider an affine approximation to the mutual information at high SNR. In particular, we investigate how mutual information varies when the numbers of antennas (at either the receiver or transmitter side) is altered. Our affine approximation to the mutual information leads to a generalization of the multiplexing gain which we call the \emph{multiplexing rate}. Such an approximation was formerly addressed in \cite{shamai}, which was the baseline of many published works, e.g.\ \cite{lozano1,jindal,Chen}.

We study the variation of the multiplexing rate when the number of antennas either at the transmit or receive side varies. More specifically, we formulate the reduction of the number of antennas by means of a convenient linear projection operator. This formulation allows us to asses the mutual information at high SNR in insightful and explicit closed form. We consider unitarily invariant matrix ensembles \cite{Percy} which model a broad class of MIMO channels \cite{tulino}. Specifically, our sole restriction is that the matrix of left (right) singular vectors of the initial channel matrix, i.e. before the reduction, is Haar distributed. Informally speaking, this implies that the channel matrix involves some symmetry with respect to the antennas. An individual antenna contributes in a ``democratic fashion" to the mutual information. There is no preferred antenna in the system. In fact, such an invariance seems a natural property for the mutual information to depend on $T$ and $R$ only, but not on the specific antennas in the system.  

Since the term $O(1)$ in \eqref{approx} is a bounded function of SNR, the expression \eqref{approx} has more than once led to misinterpretations in the wireless communications community: 
\begin{itemize}
	\item [(i)] when the number of antennas at either the transmit or receive side varies, while the minimum of the system dimensions (i.e. the numbers of transmit and receive antennas) is kept fixed, the mutual information does not vary at high SNR; 
	\item [(ii)] the mutual information scales linearly with the minimum of the system dimensions at high SNR.
\end{itemize}
It is the goal of this paper to debunk these misinterpretations. We summarize our main contributions as follows:
\begin{enumerate} 
\item  As regards misinterpretation (i) we find the following: For a system comprising $R$ receive and $T$ transmit antennas with $R>T$ ($T>R$), let some of the receive (transmit) antennas be removed from the system to obtain a system with $\tilde R\geq T$ receive ($\tilde T\geq R$ transmit) antennas. Note that $\min(T,{\tilde R})=T$ ($\min(\tilde T,{R})=R$). Then, the loss of mutual information in the high SNR limit depends only on $R$, $T$ and $\tilde R$ ($\tilde T$) and the matrix of left (right)-singular vectors of the initial $R\times T$ channel matrix, \emph{but not on its singular values}. Assuming the matrix of left-(right-)singular vectors to be Haar distributed, the ergodic rate loss is given by $\sum_{t=1}^{T}\sum_{r=\tilde R+1}^{R}\frac{1}{r-t}$ ($\sum_{r=1}^{R}\sum_{t=\tilde T+1}^{T}\frac{1}{t-r}$) nats.
\item As regards misinterpretation (ii), we quantify how the mutual information as a function of the number of antennas deviates from the approximate linear growth (versus the minimum of the system dimensions) in the high SNR limit. This deviation does depend on the singular values of the channel matrix. We show that in the large system limit the deviation is additive for compound unitarily invariant channels and can be easily expressed in terms of the S-transform (in free probability) of the limiting eigenvalue distribution (LED) of the Gramian of the channel matrix. 
\item We show that the aforementioned results on the variation of mutual information in the high SNR limit provide least upper bounds on said variation over all SNRs. Thus, these results have a universal character related to the SNR.
\item We derive novel formulations of the mutual information and the multiplexing rate in terms of the S-transform of the empirical eigenvalue distribution of the Gramian of the channel matrix. These formulations establish a fundamental relationship between the mutual information and the multiplexing rate.  
\end{enumerate}

\subsection{Related Work}
The work presented in paper \cite{lozano1} is related to contribution~1). Specifically, in \cite[Section~3]{lozano1} the authors unveiled misinterpretation (i) for iid Gaussian unitarily invariant channel matrices.  

We elucidate misinterpretation (i) by considering arbitrary unitarily invariant matrices that need neither be Gaussian nor iid. In particular, our results and/or statements do not require any assumptions on the singular values of the channel matrix. They solely depend on the singular vectors of the channel matrix, e.g. see contribution 1). Our proof technique –- which is based on an algebraic manipulation of the projection operator that we introduce –- is different from any related work we are aware of. %Actually, the projector formulation is rather essential in defining the problem per se.

\subsection{Organization}
The paper is organized as follows. In Section~\ref{not}, we introduce the preliminary notations and definitions. In Section~\ref{sysmod}, we present the system model. In Section IV, we introduce new formulations of the mutual information and the multiplexing rate in terms of the S-transform.  Section~V and VI are dedicated to lift misinterpretations (i) and (ii), respectively. Conclusions are outlined in Section~\ref{conc}. The technical lemmas and the proofs are located in the Appendix.

%%%%%%%%%%%%%%%%%
\section{Notations \& Definitions}\label{not}
\begin{notation}We denote the binary entropy function as
\begin{equation}
H(p)\triangleq{\begin{cases}
	(p-1)\log_2(1-p) -p\log_2p & \: p\in (0,1)\\
	0 & \: p \in \{0,1\} 
	\end{cases}}.
\end{equation}
\end{notation}
\begin{notation}
For an $N \times K$ matrix $\matr X$, ${\rm F}^K_{\matr X}$ denotes the empirical eigenvalue distribution function of $\matr X^\dagger\matr X$, i.e. 
\begin{equation}
{\rm F}^K_{\matr X}(x)=\frac{1}{K}\vert\left\{\lambda_i \in \mathcal L: \lambda_i { \leq }x\right\}\vert
\end{equation}
with $\mathcal L$ and $\vert\cdot \vert$ denoting the set of eigenvalues of $\matr X^\dagger \matr X$ and the cardinality of a set, respectively. Here, $(\cdot)^\dagger$ denotes conjugate transposition. Moreover, for $N,K\to\infty$ with $\phi=K/N$ fixed, if ${\rm F}^K_{\matr X}$ converges weakly and almost surely to a LED function, this limit is denoted by ${\rm F}_{\matr X}$. 
\end{notation}

\begin{definition}
A $K$-dimensional projector $\matr P_\beta$ with $\beta\leq 1$  is a $\beta K\times K$ matrix with entries $(\matr  P_\beta)_{ij} = \delta_{ij}, {\forall i,j}$, where $\delta_{ij}$ denotes the Kronecker delta. 
\end{definition}

\begin{definition}
For an $N\times K$ matrix $\matr X\ne \matr 0$, we define the normalized rank of $\matr X^\dagger\matr X$ as
\begin{equation}
\alpha_{\matr X}^K\triangleq 1-{\rm F}^{K}_{\matr X}(0) \label{alpha}
\end{equation}
and the distribution function of non-zero eigenvalues of $ \matr X^\dagger \matr X$ as
\begin{equation}
\tilde {\rm F}^{K}_{\matr X}(x)\triangleq\frac {1} {\alpha_\matr X^K} \left\{\left({\alpha_\matr X^K}-1 \right)u(x)+{\rm F}^{K}_{\matr {X}}(x)\right\} \label{newdis}
\end{equation} 
with $u(x)$ denoting the unit-step function.
\end{definition}
The S-transform introduced by Voiculescu in the context of free probability is defined as follows:
\begin{definition}\cite{berco}
Let ${\rm F}$ be a probability distribution function with support in $[0,\infty)$.  Moreover, let $\alpha\triangleq 1-{\rm F}(0)\neq 0$. Define
\begin{equation}
\Psi(z)\triangleq \int \frac{zx}{1-zx}\;{\rm dF}(x), \quad -\infty<z<0.  \label{psi}
\end{equation}
Then, the {\em S-transform} of ${\rm F}$ is defined as
\begin{equation}
{\rm S}(z)\triangleq\frac{z+1}{z}\Psi^{-1}(z), \qquad -\alpha<z<0\label{defs}
\end{equation}
where $\Psi^{-1}$ denotes the composition inverse of $\Psi$.
\end{definition}

\begin{notation}
For an $N\times K$ matrix $\matr X\ne \matr 0$, the S-transform of ${\rm F}^{K}_{\matr X}$ is denoted by ${\rm S}^{K}_{\matr X}$. For $N,K\to \infty$ with $\phi=K/N$ fixed, if $\matr X^\dagger\matr X$ has a LED function ${\rm F}_{\matr X}$ almost surely, the S-transform of ${\rm F}_{\matr X}$ is denoted by ${\rm S}_{\matr X}$. Similarly, we define $\Psi_{\matr X}^K$ and $\Psi_{\matr X}$. 
\end{notation}

All large-system limits are assumed to hold in the almost sure sense, unless explicitly stated otherwise. Where obvious, limit operators indicating the large-system limit are omitted for the sake of compactness and readability.
%%%%%%%%%%%%%%%%%%%%%%%%%%%%%%%%%%%%
{\section{System Model}\label{sysmod}}
Consider the MIMO system 
\begin{equation}
\matr y = \matr {Hx} + \matr n \label{mimo}  
\end{equation}
where $\matr H\in \CC^{R\times T}$, $\matr x\in \CC^{T\times 1}$, $\matr y\in \CC^{R\times 1}$, $\matr n\in \CC^{R\times 1}$ are respectively the channel matrix, the input vector, the output vector, and the noise vector. The entries of $\matr x$ and $\matr n$ are assumed to be independent (circularly symmetric) complex Gaussian distributed with zero mean and variances $\sigma_x^2$ and $\sigma_n^2$, respectively. The transmit SNR is defined as 
\begin{align}
\gamma&\triangleq\frac{\sigma_x^2}{\sigma_n^2}, \quad 0<\gamma <\infty.\label{SNR}  
\end{align}
The mutual information per transmit antenna of the communication link (\ref{mimo}) is given by \cite{tom}
\begin{equation}
\mathcal I(\gamma;{\rm F}_{\matr H}^T)\triangleq\int \log_2(1+\gamma x)\;{\rm dF}^T_{\matr H}(x).  \label{mut}
\end{equation}
Similarly, $\mathcal I(\gamma;{\rm F}^R_{\matr H^\dagger})$ is the mutual information per receive antenna of (\ref{mimo}). 
%%%%%%%%%%%%%%%%
\vspace{0.2cm}
\subsection{Antenna Removal Via Projector}\label{proje}
In the sequel, we formulate the variation of mutual information when the number of antennas either at the transmit or receive side of reference system \eqref{mimo} changes. This variation is achieved by removing a certain fraction of antennas at the corresponding side of the system. We formulate this removal process via a multiplication of the channel matrix with a rectangular projector matrix. 

We distinguish between two cases: the removal of receive antennas and the removal of transmit antennas. In the first case, the system model resulting after removing a fraction $1-\beta$ of receive antennas in (\ref{mimo}) reads
\begin{align}
\matr y_\beta & =  \matr P_{\beta}(\matr H\matr x + \matr n) \\
& = \matr P_{\beta}\matr H \matr x + \matr n_\beta.  \label{smodel2}
\end{align}
The $\beta R \times R$ matrix $\matr P_{\beta}$ is an $R$-dimensional projector which removes a fraction $1-\beta$ of receive antennas in reference system \eqref{mimo} and $\matr n_\beta= \matr P_{\beta}\matr n$. The mutual information of the MIMO system (\ref{smodel2}) is equal to
\begin{equation}
T\mathcal I(\gamma ;{\rm F}_{\matr P_\beta\matr H}^T).
\end{equation}
Similarly, removing a fraction $1-\beta$ of transmit antennas in (\ref{mimo}) yields the $R\times \beta T$ system
\begin{align}
\matr{\tilde y}= \matr H\matr P_{\beta}^\dagger\matr x_\beta + \matr n.\label{smodel}
\end{align}
Here, $\matr x_\beta$ is the vector obtained by removing from $\matr x$ the $(1-\beta)T$ entries fed to the removed transmit antennas, i.e. $\matr x_\beta=\matr P_\beta \matr x$ with $\matr P_\beta$ being a $T$-dimensional projector. The mutual information of system (\ref{smodel}) reads 
\begin{equation}
\beta T \mathcal I(\gamma ;{\rm F}_{\matr H\matr P_\beta^\dagger}^{\beta T}).
\end{equation}

\subsection{Unitary Invariance}\label{uniinv}
For channel matrices that are unitarily invariant from right, i.e. $\matr H$ and $\matr H\matr U$ admit the same distribution for any unitary matrix $\matr U$ independent of $\matr H$, it does not matter which transmit antennas are removed. Only their number counts. The same applies to channel matrices that are unitarily invariant from left for the removal of receive antennas. For channel matrices that involve an asymmetry with respect to the antennas, i.e.\ some antennas contribute more to the mutual information than others, it must be specified which antennas are to be removed and the mutual information will depend (typically in a complicated manner) on the choice of the removed antennas. In this paper, we restrict the considerations to cases where only the number of removed antennas matters, since this leads to explicit closed-form expressions. 

For asymmetric channel matrices, one could obtain antenna-independent scaling laws if all antennas with equal contributions to mutual information are grouped together and all those groups are decimated proportionally. Doing so would heavily complicate the formulation of the antenna removal by means of multiplication with projector matrices. However, we can utilize the fact that for the channel in \eqref{mimo}, mutual information is invariant to multiplication with unitary matrices, i.e.\
\begin{equation}
\mathcal I(\gamma ;{\rm F}^T_{\matr V\matr H\matr U}) = \mathcal I(\gamma ;{\rm F}^T_{\matr H})
\end{equation}
for all unitary matrices $\matr U$ and $\matr V$. Since the channel matrix $\matr {UHV}$ is bi-unitarily invariant for all random unitary matrices $\matr U$ and $\matr V$ independent of $\matr H$, and has the same mutual information as $\matr H$, we can assume without loss of generality that $\matr H$ is unitarily invariant from left for receive and from right for transmit antenna removal, respectively, and keep the projector formulation of Section~\ref{proje} as it is.

The multiplication with a random unitary matrix followed by a fixed selection of antennas has statistically the same effect as a random selection of antennas. It provides the symmetry required to make mutual information only depend on the number of removed antennas and not on which antennas~are~removed. 

\subsubsection*{Equivalence to the ergodic capacity variation}The ergodic capacity of channel \eqref{mimo} is \cite{Telatar}
\begin{equation}
\bar{\mathcal C}(\gamma, {\rm F}_{\matr H}^{T})\triangleq \operatorname*{\max}_{\substack{
		\matr Q\geq 0\\
		{\rm{tr}}(\matr Q)=T}}{\rm E}\left[\mathcal I(\gamma;{\rm F}_{\matr H\sqrt{\matr Q}}^T)\right]. \label{ergodiccap}
\end{equation}
Conceptually, we relax the iid assumption on the entries of $\matr x$ in \eqref{mimo} and assume arbitrary correlation between these entries described by the covariance matrix $\sigma_x^2\matr Q$ where $\matr Q$ is non-negative definite with unit trace and $\sigma_x^2\triangleq \frac{1}{T}{\rm E} [\matr x^\dagger \matr x]$. It is shown in \cite{Telatar} that for channel matrices that are unitarily invariant from right the ergodic capacity in \eqref{ergodiccap} is attained with $\matr Q={\bf I}$, i.e. 
\begin{equation}
\bar{\mathcal C}(\gamma, {\rm F}_{\matr H}^{T})={\rm E}\left[\mathcal I(\gamma;{\rm F}_{\matr H}^T)\right].
\end{equation}
In particular, the unitary invariance property of the channel is not broken by removing some of the transmit or receive antennas. For example, if $\matr H$ is invariant from right, then $\matr H\matr P_{\beta}^\dagger$ is invariant from right too. In summary, for bi-unitarily invariant channel matrices the variation of ergodic mutual informations that results from removing some number of transmit or receive antennas does actually coincide with the corresponding variation of ergodic capacities.
%It has been shown \cite{Telatar} for a channel matrix that are uniatrily invariant from right, 

%%%%%%%%%%%%%

\section{Mutual Information and Multiplexing  Rate}\label{Stransfrom}
%In this section we derive elementary expressions for the mutual information and the multiplexing rate. We also provide sufficient conditions that guarantee the convergence of these quantities when the system dimensions grow to infinity. The derivations of these results rely on recent results on the S-transform reported in \cite{hager}.

The normalized mutual information in (\ref{mut}) can be decomposed as
\begin{align}
\mathcal I(\gamma;{\rm F}_{\matr H}^T)&=\underbrace{\alpha_\matr H^T \int \log_2(\gamma x)\;{\rm d}{\tilde{\rm F}}^T_{\matr H}(x)}_{\mathcal I_{0}(\gamma;{\rm F}_{\matr H}^T)}\nonumber \\ &\quad +\underbrace{\alpha_{\matr H}^T\int \log_2\left(1+\frac{1}{x\gamma}\right)\;{\rm d}{\tilde{\rm F}}^T_{\matr H}(x)}_{\Delta\mathcal{I}(\gamma;{\rm F}^T_{\matr {H}})}. \label{dev}  
\end{align}
We refer to the first term $\mathcal I_{0}(\gamma;{\rm F}^T_{\matr H})$ as the {\em multiplexing rate} per transmit antenna. The factor $\alpha_\matr H^T$ is the multiplexing gain normalized by the number of transmit antennas. The second term $\Delta\mathcal{I}(\gamma;{\rm F}^T_{\matr H})$ is the difference between the mutual information per transmit antenna and the multiplexing rate per transmit antenna. To alleviate the terminology, in the sequel we skip the explicit reference to the normalization by the number of transmit (or receive, see later) antennas when we refer to quantities such as those arising in \eqref{dev}. Whether the quantities considered are absolute or normalized will be clear from the context. We have
\begin{eqnarray}
\lim_{\gamma \to \infty}\Delta\mathcal{I}(\gamma;{\rm F}^T_{\matr H})=0. \label{deltadif}
\end{eqnarray}
If $\matr H^\dagger\matr H$ is invertible we have
\begin{align}
\mathcal{I}_0(\gamma;{\rm F}^T_{\matr {H}})&=\frac{1}{T} \log_2\det\left(\gamma\matr H^\dagger\matr H\right)\\
\Delta\mathcal{I}(\gamma;{\rm F}^T_{\matr {H}}) &= \frac{1}{T}\log_2\det\left(\bf I+{(\gamma\matr H^\dagger\matr H)^{-1}}\right)
\end{align}
with $\bf I$ denoting the identity matrix.  

The affine approximation of the ergodic mutual information at high SNR introduced in \cite{shamai}, see also \cite[Eq. (9)]{lozano1} for a compact formulation of it, coincides with the ergodic formulation of our definition of the multiplexing rate.

We next uncover a fundamental link between the mutual information and the multiplexing rate. This result makes use of the minimum-mean-square-error (MMSE) achieved by the optimal receiver for \eqref{mimo} normalized by the number of transmit antennas
\begin{equation}
\eta_{\matr H}^T(\gamma)\triangleq \int \frac{{\rm dF}^T_{\matr H}(x)}{1+\gamma x}. \label{eta}
\end{equation}
Clearly, $\eta_{\matr H}^T(\gamma)$ is a strictly decreasing function of $\gamma$ with range $(1-\alpha_{\matr H}^T,1)$ \cite{tulino}. 
\begin{theorem}\label{main}
Define
\begin{equation}
f_{\matr H}(x)\triangleq H(x)-\int_{0}^{x}\log_2{\rm S}^T_{\matr H}(-z)\;{\rm d}z, \quad 0\leq x\leq \alpha_{\matr H}^T. \label{goodform}
\end{equation}
Then, we have
\begin{align}
\mathcal I(\gamma;{\rm F}^T_{\matr H})&=f_{\matr H}(1-\eta^T_{\matr H})+(1-\eta_{\matr H}^T)\log_2\gamma\label{etaalpha} \\
\mathcal I_{0}(\gamma;{\rm F}^T_{\matr H})&= f_{\matr H}(\alpha^T_{\matr H})+\alpha_{\matr H}^T \log_2\gamma.\label{mult}
\end{align}
For short we write $\eta_{\matr H}^T$ for $\eta_{\matr H}^T(\gamma)$ in \eqref{etaalpha}.
\proof{See Appendix \ref{Lemma4new}.}
\end{theorem}
Note that by definition the function $f_{\matr H}(x)$ in \eqref{goodform} may involve $\alpha_{\matr H}^T$ via ${\rm S}_{\matr H}^{{T}}(z)$. We have the following implications of Theorem~\ref{main}: i) the mutual information can be directly expressed as a function of the (normalized) MMSE; ii) for any expression of the mutual information as a function of the MMSE $\eta_{\matr H}^T$ the multiplexing rate results immediately by substituting $\eta_{\matr H}^T$ for $1-\alpha_{\matr H}^T$, e.g. see Examples~\ref{Example1} and \ref{Example2}; iii) the converse of ii) is not always true: given an expression of the multiplexing rate as a function of $\alpha_{\matr H}^T$, substituting $\alpha_{\matr H}^T$ for $1-\eta_{\matr H}^T$ does not always yield the mutual information. An intermediate step is required here to guarantee that the converse holds: the expression needs first to be recast as a function of $f_\matr H$. Then substituting $\alpha_{\matr H}^T$ for $1-\eta_{\matr H}^T$ in the latter function yields the mutual information.

If any probability distribution function with support in $[0,\infty)$, say $\rm F$, is substituted for ${\rm F}_{\matr H}^T$ in \eqref{dev} the formulas \eqref{etaalpha} and \eqref{mult} remain valid provided $\mathcal I(\gamma;{\rm F})$ is finite and $\log(x)$ is absolutely integrable over ${\rm \tilde F}$, respectively\footnote{Here $\tilde{\rm F}$ is defined by substituting ${\rm \tilde F}_{\matr H}^T$ for ${\rm F}$ in \eqref{newdis}.}. The absolute integrability condition holds if, and only if, $\mathcal I(\gamma;{\rm F})$ and $\Delta\mathcal I(\gamma;{\rm F})$ are finite, see \eqref{L9}-\eqref{10b}. In the sequel we substitute ${\rm F}_{\matr H}$ for ${\rm F}^T_{\matr H}$ to calculate $\mathcal I(\gamma;{\rm F}_{\matr H})$ and $\mathcal I_0(\gamma;{\rm F}_{\matr H})$. In Appendix~C, we provide some sufficient conditions that guarantee the almost sure convergence of  $\mathcal I(\gamma;{\rm F}^T_{\matr H})$ and $\mathcal I_0(\gamma;{\rm F}^T_{\matr H})$ to  $\mathcal I(\gamma;{\rm F}_{\matr H})$ and $\mathcal I_0(\gamma;{\rm F}_{\matr H})$, respectively.  We conclude that these asymptotic convergence are reasonable assumptions in practice, for the details see Appendix~C.

It is well-known that the S-transform of the LED of the product of asymptotically free matrices is the product of the respective S-transforms of the LEDs of these matrices. Therefore, for MIMO channel matrices that involve a compound structure, Theorem~\ref{main} provides a means to \emph{analytically} calculate the large-system limits of the mutual information and multiplexing rate in terms of the large-system limits of the \emph{MMSE} and the \emph{multiplexing gain}. We next address two relevant random matrix ensembles that share this structure.

\begin{example}\label{Example1}
We consider the concatenation of vector-valued fading channels described in \cite{ralfa}. Specifically, we assume that the channel matrix $\mathbf{\matr H}$ factorizes according to
\begin{equation}
\matr{\matr H}=\matr{X}_N\matr{X}_{N-1}\cdots\matr{X}_2\matr{X}_1\label{matrixx}
\end{equation}
where the entries of the $K_{n}\times K_{n-1}$ matrix $\matr{X}_n $ are iid with zero mean and variance $1/K_n$ for $n\in [1,N]$. Furthermore, the ratios $\rho_n\triangleq K_n/K_0$  $n\in [1,N]$ are fixed as $K_n\to \infty$. Moreover, let $\eta_\matr H$ denote the large-system limit of MMSE $\eta_\matr H^T$. By invoking Theorem~\ref{main} we obtain an analytical expression of the large-system limit of the mutual information in terms of (the large-system limit of) the MMSE\footnote{An explicit expression of the MMSE as a function of SNR is difficult to obtain. However, $\eta_{\matr H}(\gamma)$ can be solved numerically from the fixed point equation $\gamma=\frac{\eta_{\matr H}(\gamma)}{1-\eta_{\matr H}(\gamma)}\prod_{n=1}^{N}\frac{\eta_{\matr H}(\gamma)+\rho_n-1}{\rho_n}$ \cite[Eq.~(21)]{ralfa}.} as 
\begin{align}
&\mathcal I(\gamma;{\rm F}_{\matr H})= H(\eta_\matr H)+(1-\eta_{\matr H})(\log_2\gamma-N\log_2e)\nonumber \\
&+(1-\eta_{\matr H})\left[\sum_{n=1}^{N}\frac{\rho_n}{1-\eta_\matr H} H\left(\frac{1-\eta_\matr H}{\rho_n} \right)+\log_2 \frac{1-\eta_\matr H}{\rho_n}\right] \label{iidmut}.
\end{align}
Furthermore, as regards the multiplexing rate, we have 
\begin{align}
\mathcal I_{0}(\gamma;{\rm F}_{\matr H})=& H(\alpha_\matr H)+\alpha_{\matr H}(\log_2\gamma-N\log_2e) \nonumber \\
& +\alpha_{\matr H}\left[\sum_{n=1}^{N}\frac{\rho_n}{\alpha_\matr H} H\left(\frac{\alpha_\matr H}{\rho_n} \right)+\log_2 \frac{\alpha_\matr H}{\rho_n} \right]\label{iidmultip}
\end{align}
with $\alpha_{\matr H}=\min (1,\rho_1,\cdots, \rho_N)$.
\proof{See Appendix~\ref{PExample1}.}
\end{example}

\begin{example}\label{Example2}
We consider a Jacobi matrix ensemble, see e.g. \cite{Alain}, \cite{Edelman}, which find application in the context of optical MIMO communications \cite{Ronen},\cite{Aris14}. Accordingly, the channel matrix factorizes as
\begin{equation}
\matr H = \matr P_{\beta_2} \matr U \matr P_{\beta_1}^{\dagger} \label{jacobi}
\end{equation}
where $\matr U$ is an $N\times N$ Haar unitary matrix. From Theorem~\ref{main} we obtain 
\begin{align}
\mathcal I(\gamma;{\rm F}_{\matr H})=& H(\eta_\matr H)+(1-\eta_{\matr H})\log_2 \gamma \nonumber \\
&-\frac{H(\beta_1(1-\eta_{\matr H}))}{\beta_1}+\frac{\beta_2}{\beta_1}H\left( \frac{\beta_1}{\beta_2}(1-\eta_{\matr H})\right) \label{jacobimut}
\end{align}
where $\eta_{\matr H}=\eta_{\matr H}(\gamma)$ is given by
\begin{equation}
\eta_\matr H(\gamma)=1+\frac{-(1+\kappa\gamma)+\sqrt{(1+\kappa\gamma)^2-4\beta_1\beta_2\gamma(1+\gamma)}}{2\beta_1(1+\gamma)}  \label{etajacobi}
\end{equation}
with $\kappa\triangleq\beta_1+\beta_2$. Moreover, we have 
\begin{align}
\mathcal I_0(\gamma;{\rm F}_{\matr H})=& H(\alpha_{\matr H})+\alpha_{\matr H}\log_2 \gamma \nonumber \\
&-\frac{H(\beta_1\alpha_{\matr H})}{\beta_1}+\frac{\beta_2}{\beta_1}H\left( \frac{\beta_1}{\beta_2}\alpha_{\matr H}\right) \label{jocobimultip}
\end{align}
with $\alpha_{\matr H}=\min (1,\beta_2/\beta_1)$. 
\proof{See Appendix~\ref{PExample2}.}
\end{example}
 
 \section{The Universal Rate Loss}
In Section~1 we underlined the following misinterpretation of mutual information: \emph{when the number of antennas (at either the transmit or receive side) varies, with the minimum of the system dimensions kept fixed, the mutual information does not vary at high SNR.} It is the goal of this section to elucidate this misinterpretation. To do so we need to distinguish between two cases as to reference system \eqref{mimo}: (i) $T\leq R$; (ii) $T \geq R$. In the former (latter) case we consider the removal of receive (transmit) antennas. In both cases the reduction of antennas is constrained in a way that keeps the minimum of the numbers of antennas at both sides fixed.
\subsection{Case~(i) -- Removing receive antennas} 
We remove a fraction $(1-\beta)$ of receive antennas in system \eqref{mimo} to obtain system \eqref{smodel2}. We constrain the reduction with the condition $\beta\geq \phi\triangleq T/R$ to ensure that $\min (T,\beta R)=T$. This reduction of the number of receive antennas causes a loss in mutual information given by $T\mathcal I(\gamma;{\rm F}^T_{\matr H})-T\mathcal I(\gamma;{\rm F}^T_{\matr P_{\beta}\matr H})$. Normalizing this loss with the number of transmit antennas yields 
\begin{equation}
	\mathcal I(\gamma; {\rm F}_{\matr H}^T) -\mathcal I(\gamma;{\rm F}^T_{\matr P_{\beta}\matr H}). \label{receiveloss} 
\end{equation}
Assume that $\matr H$ and $\matr P_{\beta}\matr H$ have both full rank almost surely. Then, we define the rate loss
\begin{align}
	\chi_{\matr H}^T(R,\beta R)&\triangleq \lim_{\gamma\to \infty}\mathcal I(\gamma; {\rm F}_{\matr H}^T) -\mathcal I(\gamma;{\rm F}^T_{\matr P_{\beta}\matr H}),~\beta \geq \phi. \label{rateloss} \\ &=\mathcal I_0(\gamma; {\rm F}_{\matr H}^T) -\mathcal I_0(\gamma;{\rm F}^T_{\matr P_{\beta}\matr H}) \\ &= \frac{1}{T}\log_2 \frac{\det\matr H^\dagger \matr H}{\det\matr H^\dagger \matr P_{\beta}^\dagger \matr P_\beta \matr H}.	\label{ratef}
\end{align}
The full-rank assumption implies $\alpha^T_{\matr H}=\alpha^T_{\matr P_{\beta}\matr H}$ which is essential in the definition \eqref{rateloss}. Otherwise the difference in \eqref{rateloss} diverges as $\gamma \to \infty$. Next, we present some general important properties of the rate loss $\chi_{\matr H}^T(R,\beta R)$.
\subsubsection{Universality related to SNR}
Note that both quantities in \eqref{receiveloss} increase with the SNR. It is shown in Appendix~\ref{PSupI} that their difference, i.e. \eqref{receiveloss}, increases with the SNR too. Hence, the rate loss $\chi_{\matr H}^T(R,\beta R)$ provides the least upper bound on the mutual information loss over the entire SNR range. 
\begin{remark}\label{SupI}
Let $\matr H$ and $\matr P_{\beta}\matr H$ have both full rank almost surely. Then, we have
\begin{align}
\chi_{\matr H}^T(R,\beta R)&= \sup_{\gamma}\{\mathcal I(\gamma; {\rm F}_{\matr H}^T) -\mathcal I(\gamma;{\rm F}^T_{\matr P_{\beta}\matr H})\}.
\end{align}
\proof{See Appendix~\ref{PSupI}}
\end{remark}
\subsubsection{Equivalence to capacity loss}\label{caprel} Let us denote the capacity of channel \eqref{smodel2} as
 \begin{equation}
 {\mathcal C}(\gamma; {\rm F}_{\matr P_{\beta}\matr H}^{T})\triangleq \operatorname*{\max}_{\substack{
 		\matr Q\geq 0\\
 		{\rm{tr}}(\matr Q)=T}}\mathcal I(\gamma;{\rm F}_{\matr P_{\beta}\matr H\sqrt{\matr Q}}^T). \label{ccap}
 \end{equation}
It turns out that \eqref{rateloss} also holds when the mutual informations in \eqref{rateloss} are replaced by the respective capacities.
\begin{remark}\label{Cap}
Let $\matr H$ and $\matr P_{\beta}\matr H$ have both full rank almost surely. Then, we have
	\begin{equation} 
	\chi_{\matr H}^T(R,\beta R)=\lim_{\gamma\to \infty}\mathcal C(\gamma; {\rm F}_{\matr H}^T) -\mathcal C(\gamma;{\rm F}_{\matr P_{\beta}\matr H}^T).
	\end{equation}
\proof{See Appendix~\ref{PCap}.}
\end{remark}
\subsubsection{The invariance related to singular values}
Though $\chi_{\matr H}^T(R,\beta R)$ is defined through the distribution functions ${\rm F}^T_{\matr H}$ and ${\rm F}^T_{\matr P_{\beta}\matr H}$ in \eqref{rateloss}, it actually depends solely on the matrix of left singular vectors of $\matr H$:
\begin{theorem}\label{Theorem2}
Let $\matr H$ and $\matr P_{\beta}\matr H$ have both full rank almost surely. Consider the spectral decomposition 
\begin{equation}
\matr H=\matr L \matr S \matr R \label{SVD}
\end{equation}
where $\matr L$ is a $R\times R$ unitary matrix whose columns are the left singular vectors of $\matr H$, $\matr R$ is a $T\times T$ unitary matrix whose columns are the right singular vectors of $\matr H$ and the diagonal entries of $\matr S$ are the singular values of $\matr H$. Then, we have
\begin{equation}
\chi_{\matr H}^T(R,\beta R)=-\frac{1}{T}\log_2\det \matr P_{\phi}\matr L^\dagger \matr P_{\beta}^\dagger \matr P_{\beta} \matr L \matr P_{\phi}^\dagger. \label{result1}
\end{equation}
\proof{See Appendix~\ref{PTheorem2}.}
\end{theorem}
\subsubsection{Statistical properties resulting from unitarily invariance}
Let $\matr H^\dagger \matr H$ have full rank almost surely and be unitarily invariant\footnote{
Provided $\matr H^\dagger \matr H$ is unitarily invariant, when $\matr H$ has almost surely full rank, so does $\matr P_{\beta} \matr H$ too, see Appendix~\ref{PTheorem3}.}. Thereby, the matrix of left singular vectors of $\matr H$, i.e. $\matr L$, is  Haar, see \cite[Lemma 2.6]{tulino}. Thus, $\matr P_{\phi}\matr L^\dagger \matr P_{\beta}^\dagger \matr P_{\beta} \matr L \matr P_{\phi}^\dagger$ belongs to the Jacobi matrix ensemble, see Example~\ref{Example2}. In other words, the rate loss $\chi_{\matr H}^T(R,\beta R)$ becomes nothing but minus the $\log\det$ of the Jacobi matrix ensemble normalized by $T$. We also refer the reader to \cite{Alain} for a detailed study of the determinant of the Jacobi matrix ensemble. In particular, from \cite[Proposition~2.4]{Alain}, the rate loss admits the explicit statistical characterization
\begin{equation}
	\chi_{\matr H}^T(R,\beta R)\sim-\frac{1}{T}\sum_{t=1}^{T}\log_2 \rho_t \label{Betacar}
\end{equation}
where $\{\rho_1,\cdots,\rho_T\}$ are independent random variables and $\rho_t \sim ~ {\mathcal Be}\left((\beta R+1-t),(1-\beta)R\right)$. Here, $X\sim Y$ indicates that random variables $X$ and $Y$ are identically distributed. For $a>0$ and $b>0$,  ${\mathcal Be}(a,b)$ denotes the Beta distribution with density 
\begin{equation}
	{\mathcal Be}(x;a, b)=\frac{\Gamma (a+b)}{\Gamma(a)\Gamma (b)}x^{a-1}(1-x)^{b-1},\quad x>0
\end{equation}
where $\Gamma$ is the gamma function. 

\begin{corollary}[Universal Rate Loss]\label{Theorem3}
Let $\matr H^\dagger \matr H$ have full rank almost surely and be unitarily invariant. Define\footnote{The sum over an empty index set is by definition zero.}
\begin{equation} 
\chi^T(R,R')\triangleq\frac{1}{T \ln 2}\sum_{t=1}^{T}\sum_{r=R'+1}^{R}\frac{1}{r-t}\ ,\quad T\leq R'\leq R.\label{rtlos}
\end{equation}
Then, we have 
\begin{equation}
{\rm E}[\chi_{\matr H}^T(R,\beta R)]=\chi^T(R,\beta R). \label{ergodic}
\end{equation}
Moreover, if $R,T\to \infty$ with $\phi=T/R$ fixed, we have almost surely
\begin{equation}
\chi_{\matr H}^T(R,\beta R)\to\frac {H\left(\phi\right)}\phi- \frac\beta\phi H\left(\frac\phi\beta\right). \label{Loss}
\end{equation}

\proof{See Appendix~\ref{PTheorem3}.}
\end{corollary}

The name Universal Rate Loss refers to the fact that the results in Corollary~1 solely refer to the number of transmit and receive antennas before and after the variation. The ergodic rate loss has the additive property
\begin{equation}
\chi^T(R,R')=\chi^T(R,T)-\chi^T(R',T)\ ,\quad T\leq  R'\leq R\label{linearprop}.
\end{equation}
Note that $\chi^T(R,T)$ equals to the ergodic rate loss when we remove as many antennas as needed to obtain a square system. Furthermore, if $R,T\to \infty$ with the ratio $\phi=T/R$ fixed, the first and the second terms of \eqref{linearprop} converge to respectively the first and the second terms of \eqref{Loss}. 

We coin the limit \eqref{Loss} the \emph{binary entropy loss} as it only involves the binary entropy function evaluated at the aspect ratios $\phi$ and $\beta/\phi$ of two channel matrices -- the one before and the one after the removal of the antennas. In particular, for $\beta=\phi$, i.e. we remove as many receive antennas as needed to obtain a square system, the binary entropy loss has the compact expression $H(\phi)/\phi$. 
\subsubsection{A symmetry property of the universal rate loss} 
We show a symmetry property of the universal rate loss in the case when the end system after (completion of the antenna removal) is square, i.e. $\beta=\phi$. Let us start with an illustrative example. Consider two separate MIMO systems one of dimensions $3\times 2$ and one of dimensions $3\times 1$. Let the antenna removal  processes be $3\times 1 \to 1\times 1$ for the former system and $3\times 2 \to 2\times 2$ for the latter. Thus, in both cases two communication links are removed from the reference systems. Let the channel matrices of the reference systems fulfill the conditions stated in Corollary~\ref{Theorem3} (i.e. full-rank and unitary invariance). Both removal process lead to the same the binary entropy loss equal to $3H(1/3)=3H(2/3)=2.75$ bit.
\begin{remark}\label{symmetry}
The function $T\chi^{T}(R,R')$ (see \eqref{rtlos}) satisfies the symmetry property
\begin{equation}
T\chi^T(R,T)=T'\chi^{T'}(R,T')\ , \quad T<R \label{symm}
\end{equation}
where $T'\triangleq R-T$.
\proof{See Appendix~\ref{Psymmetry}}
\end{remark}
Note that the expressions $T\chi^T(R,T)$ and $T'\chi^{T'}(R,T')$ corresponds to the ergodic rate losses for the antenna removal processes $R\times T\to T\times T$ and $R\times T'\to T'\times T'$, respectively. In both cases $T\times (R-T)$ communications links are removed from the reference systems. In other words, for $R$ being fixed the ergodic rate loss $T\chi^T(R,T)$ is a symmetric function of $T$ with respect to $T=R/2$ (see Figure~2).

Since $\chi^{\phi R}(R,\beta R)\leq \chi^{\phi R}(R,\phi R)$, the symmetry property \eqref{symm} implies that the maximum ergodic rate loss is attained when $\phi=\beta=1/2$.
\begin{remark}
Let $\matr H^\dagger \matr H$ have full rank almost surely and be unitarily invariant. Then, for $\phi \leq \beta \leq 1$ we have 
\begin{equation}
\left(\frac{1}{2},\frac{1}{2}\right)= \arg \max_{\phi, \beta} {\rm E}[\chi_{\matr H}^{\phi R}(R,\beta R)].
\end{equation}
\end{remark}

\subsection{Case~(ii) -- Removing transmit antennas} 
We remove a fraction $(1-\beta)$ of transmit antennas in \eqref{mimo} to obtain system \eqref{smodel}. We constrain the reduction of receive antennas with $\beta\geq 1/\phi$ ($\phi=T/R$) to ensure $\min (\beta T, R)=T$. Reducing the number of transmit antennas results in a loss of mutual information equal to $T\mathcal I(\gamma;{\rm F}^T_{\matr H})-\beta T\mathcal I(\gamma;{\rm F}^{\beta T}_{\matr H\matr P_{\beta}^\dagger})$. Normalizing this loss with the number of transmit antennas of the reference system gives 
\begin{equation}
\mathcal I(\gamma; {\rm F}^{T}_{\matr H}) -{\beta}\mathcal I(\gamma;{\rm F}^{\beta T}_{\matr H\matr P_\beta^\dagger}). \label{transmitloss}
\end{equation}
Let $\matr H$ and $\matr H\matr P_{\beta}^\dagger$ have both full rank almost surely. Then, we define the large SNR limit
\begin{equation}
\tilde\chi_{\matr H}^R(T,\beta T)\triangleq \lim_{\gamma\to \infty}\mathcal I(\gamma; {\rm F}^T_{\matr H})-\beta\mathcal I(\gamma;{\rm F}^{\beta T}_{\matr H\matr P_\beta^\dagger})\ , \quad \beta \geq \frac{1}{\phi}. \label{rateloss2}
\end{equation}
Again the full rank assumption is important here. Otherwise the difference \eqref{rateloss2} may diverge as $\gamma \to \infty$.
\begin{corollary}\label{Theorem2b}
Let $\matr H$ and $\matr H\matr P_{\beta}^\dagger$ have both full rank almost surely. Then, we have
	\begin{equation}
	\tilde\chi_{\matr H}^R(T,\beta T)=-\frac{1}{T}\log_2\det \matr P_{\frac 1 \phi}\matr R\matr P_{\beta}^\dagger \matr P_{\beta} \matr R^\dagger \matr P_{\frac 1 \phi}^\dagger \label{result2}
	\end{equation}
where $\matr R$ is a $T\times T$ unitary matrix whose columns are the right singular vectors of $\matr H$, see \eqref{SVD}. 
\end{corollary}
Note that the right-hand side in (\ref{result2}) is obtained by formally replacing $\phi$ with $\phi^{-1}$ in the right-hand side of \eqref{Loss}. This follows from the identity
\begin{equation}
\beta\mathcal I(\gamma;{\rm F}_{\matr H\matr P_\beta^\dagger}^{\beta T})=\frac{1}{\phi}\mathcal I(\gamma;{\rm F}_{\matr P_\beta\matr H^\dagger}^{R}).
\end{equation}
This substitution is valid for any result that refers to mutual information, e.g. as in\ Corollary~\ref{Theorem3}. However, it does not apply in general to capacity related results, such as in Remark~\ref{Cap}, due to the placement of the projection operator on the transmitter side. 
\subsection{The rate loss with antenna power profile}
In this subsection we address the rate loss $\chi_{\matr H}^T$ for a channel model that takes into consideration the power imbalance at the transmitter and receiver sides:
\begin{equation}
\matr H={{\matr \Lambda_{\rm R}}}\matr {\tilde H}{\matr \Lambda_{\rm T}}.\label{gmodel}
\end{equation}
Here, the matrices $\matr \Lambda_{\rm R}\in \CC^{R\times R}$ and $\matr \Lambda_{\rm T}\in \CC^{T\times T}$ are diagonal, full-rank, and deterministic. The matrix $\matr \Lambda_{\rm R}$ ($\matr \Lambda_{\rm T}$) represents the power imbalance at receive (transmit) side. 

We generalize Theorem~\ref{Theorem2} for the model \eqref{gmodel} as (see Appendix~\ref{PTheorem2})
\begin{equation}
\chi_{\matr H}^T(R,\beta R)=\frac{1}{T}\log_2\frac{\det \matr P_{\phi}\matr {\tilde L}^\dagger \Theta_1 \matr {\tilde L} \matr P_{\phi}^\dagger}{\det \matr P_{\phi}\matr {\tilde L}^\dagger \Theta_\beta \matr {\tilde L} \matr P_{\phi}^\dagger}\label{general}
\end{equation}
where $\Theta_\beta\triangleq{{{\matr \Lambda_{\rm R}}}}^\dagger \matr P_{\beta}^\dagger\matr P_{\beta}{{{\matr \Lambda_{\rm R}}}}$ for $\beta\leq 1$ and $\matr {\tilde L}$ is a $R\times R$ unitary matrix whose columns are the left singular vectors of $\matr {\tilde H}$, see \eqref{SVD}. Note that the rate loss does not depend on the singular values of $\matr {\tilde H}$. This property allows for obtaining a convenient expression for the ergodic rate loss ${\rm E}[\chi_{\matr H}^T(R,\beta R)]$
when $\matr {\tilde H}^\dagger\matr {\tilde H}$ is unitarily invariant, i.e. $\matr {\tilde L}$ in \eqref{general} is Haar distributed. 
\begin{corollary}\label{Cor3}
Let $\matr H$ be defined as in \eqref{gmodel}. Furthermore, let $\matr {\tilde H}^\dagger\matr {\tilde H}$ have full rank almost surely and be unitarily invariant. Moreover, let $\matr X_{\beta}\triangleq \matr P_{\beta}\matr X$ where $\matr X$ is a $R\times T$ matrix with iid zero-mean complex Gaussian entries. Let $\matr {D}_{\beta}$ be the $\beta R\times \beta R$ diagonal matrix whose diagonal entries are the non-zero eigenvalues of $\Theta_\beta\triangleq{{{\matr \Lambda_{\rm R}}}}^\dagger \matr P_{\beta}^\dagger\matr P_{\beta}{{{\matr \Lambda_{\rm R}}}}$. Then, we have 
\begin{equation}
{\rm E}[\chi_{\matr H}^T(R,\beta R)]=\frac{1}{T}{\rm E}\left[\log_2\frac{\det \matr X_1^\dagger \matr {D}_{1}\matr X_1}{\det \matr X_{\beta}^\dagger \matr {D}_{\beta}\matr X_\beta}\right]. \label{EE}
\end{equation} 
\proof{See Appendix~\ref{PCor3}}
\end{corollary}
The expectation in \eqref{EE} can be simply computed by using the following result.
\begin{lemma}\cite[Lemma~2]{lozano1}
Let $\matr X$ be an $n\times m$ matrix with iid zero-mean complex Gaussian entries such that $n>m$. Let $\matr D$ be an $n\times n$ deterministic Hermitian positive-definite matrix whose $j$th eigenvalue is denoted by $\lambda_{j}$. Moreover, let $\matr \Omega$ be the $n\times n$ Vandermonde matrix with $(\matr \Omega)_{ij}=\lambda_i^{j-1}$ and $\matr \Gamma$ be the $(n-m)\times (n-m)$ principal submatrix of $\matr \Omega$. Then, we have 
\begin{equation}
{\rm E}[\ln \det \matr X^\dagger \matr D\matr X]=\frac{\det \matr \Gamma}{\det\matr \Omega} \sum_{i=1}^{m} \det \matr \Psi_{i}
\end{equation}
where $\matr \Psi_{i}$ is $m\times m$ matrix whose entries are
\begin{align}
(\matr \Psi_i)_{k,l}=&\nu_{n-m+k}\lambda_{n-m+k}^{n-m-1+l}\nonumber \\&-\sum_{d=1,q=1}^{n-m}\nu_{q}(\matr \Gamma^{-1})_{d,q}^{d-1}\lambda_{n-m+k}^{d-1}\lambda_{q}^{n-m-1+1}.
\end{align}
In this expression, $\nu_{q}=\psi(l)+\ln\lambda_{q}$ if $l=i$ else $\nu_{q}=1$ with $\psi(\cdot)$ denoting the digamma function. 
\end{lemma} 
\subsection{Further discussions based on numerical results}
As a warm up example, consider a $4\times 2$ MIMO system that is stripped off two of its four receive antennas. For full-rank channel matrices that are unitarily invariant from left Theorem~1 gives the exact high SNR limit of the ergodic loss equal to $4\chi^2(4,2)=3.37$ bit. The asymptotic loss \eqref{Loss} is $4H(2/4)=4$ bit. 
\begin{figure}
\epsfig{file=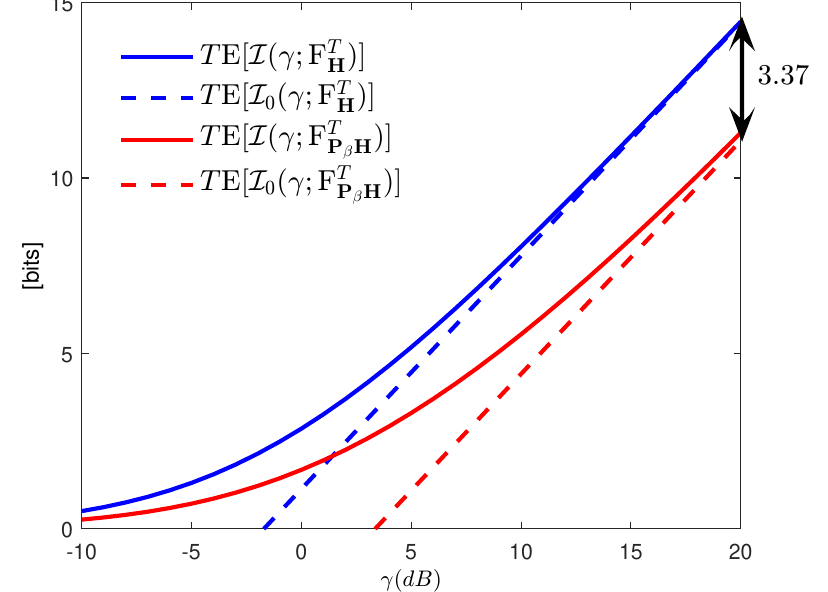,width=\columnwidth}
\caption{Ergodic mutual information (continuous lines) and ergodic multiplexing rate (dashed lines) versus the SNR of a zero-mean iid complex Gaussian MIMO channel with $T = 2$ transmit antennas and the number of receive antennas decreased from $R=4$ (blue curves) to $R=2$ (red curves).}\label{fig}
\end{figure} 
Note also that $4\chi^{2}(4,2)$ is the supremum of the mutual information loss over all SNRs. This is depicted in Figure~1 for a Gaussian channel. %From Figure~1 we may conclude that $\chi^{2}(4,2)$ yields an accurate approximation of mutual information loss when the SNR is around $20$dB and above. 

We illustrate
\begin{figure}
	\epsfig{file=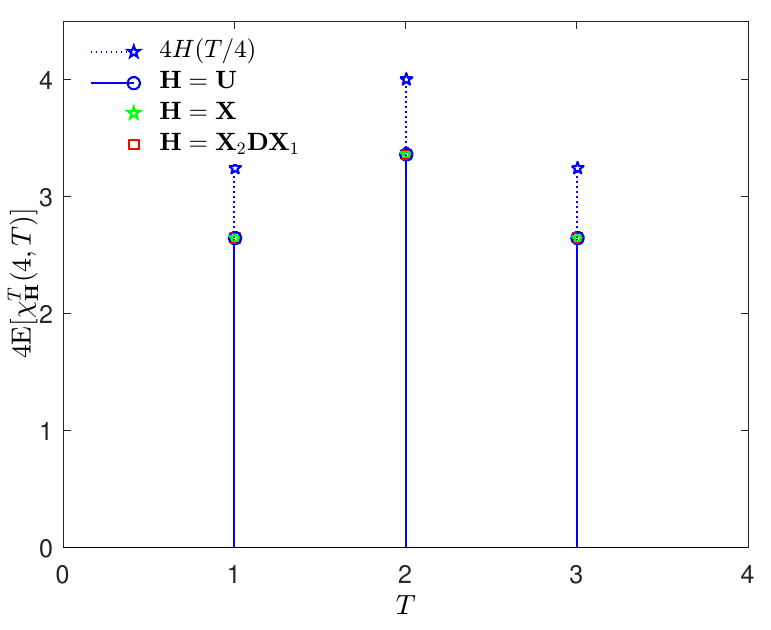,width=\columnwidth}
	\caption{The maximal ergodic mutual information loss over the SNR range: The entries of $\matr X\in \CC^{4\times T}$, $\matr X_1\in \CC^{S=4\times T}$ and $\matr X_2\in \CC^{4\times S=4}$ are zero mean iid complex Gaussian. The matrix $\matr U\in \CC^{4\times T}$ is uniformly distributed over the manifold of complex $4\times T$ matrices. The $S\times S$ matrix $\matr D$ is positive diagonal. Its diagonal entries are iid and uniformly distributed.}\label{fig2}
\end{figure}the universal rate loss and the tightness of the approximation provided by the binary entropy loss, i.e. $RH(T/R)$, already for small system dimensions. To this end we consider three different channel models that are unitarily invariant from the left: (i) the channel matrix $\matr H=\matr U {\matr \Lambda}$ where $\matr U\in \CC^{R\times T}$ is uniformly distributed over the manifold of complex $R\times T$ matrices such that $\matr U^\dagger \matr U=\matr I$ and $\matr \Lambda \in \RR^{T\times T}$ is a positive diagonal matrix that represents the power imbalance at the transmitter. This is a typical channel model in the context of massive MIMO, i.e. in the regime of $T\ll R$.  Here we point out that $\matr\Lambda$ does not affect the rate loss. Therefore for convenience we set $\matr \Lambda=\matr{\bf I}$. (ii) the channel matrix $\matr H=\matr X$ with the entries of $\matr X$ being zero-mean iid complex-valued Gaussian with finite variance; (iii) the channel matrix $\matr H=\matr X_2\matr D \matr X_1$. Here $\matr X_1\in \CC^{S\times T}$ and $\matr X_2\in \CC^{R\times S}$ represent the propagation channel from the transmit antennas to the scatterers and from the scatterers to the receive antennas respectively, while the diagonal entries in the diagonal matrix $\matr D$ are the individual scattering coefficients of the scatterers. This random matrix ensemble models the channel under the assumption of propagation via one-bounce scattering only \cite{ralfb}. To fulfill the full-rank condition we restrict to the case $S\geq T$. From Figure~2, we conclude that the binary entropy loss yields an accurate approximation even for small system dimensions. 
 
\section{Deviation from Linear Growth}
In this section we clarify the second misinterpretation underlined in Section~1. Specifically, we analyze the variation of the multiplexing rate when either the number of receive or the number transmit antennas varies while their maximum is kept fixed. 

For a channel matrix having \emph{orthogonal} columns when the number of transmit or receive antennas varies, the linear growth of mutual information is obvious. However, for a channel matrix with e.g. iid entries, a substantial crosstalk arises due to the lack of orthogonality of its columns. The effect of this crosstalk onto mutual information is non-linear in the number of antennas.

The mutual information scales approximately linearly in the minimum of the numbers of transmit and receive antennas. For a tall rectangular channel matrix that becomes wider and wider, the mutual information can only grow approximately linearly until the matrix becomes square. The same holds for a wide rectangular channel matrix growing taller and taller. Therefore, we have to distinguish between two cases: (i) the number of receive antennas is smaller than the number of transmit antennas, i.e.\ a wide channel matrix, and (ii) the converse of (i), i.e.\ a tall channel matrix. Since case (ii) can be easily treated by replacing the channel matrix with its conjugate transpose, we restrict our investigations to case (i).

The linear growth cannot continue once the channel matrix has grown square. Thus, it makes sense to constrain the matrix of reference system (\ref{mimo}) to be square; i.e. \textit{we assume that the channel matrix $\matr H$ in \eqref{mimo} is $N\times N$ i.e. $N=R=T$}.

The exact mutual information of the (rectangular) system \eqref{smodel} of size $\beta N \times N$, $\beta\leq 1$ is
\begin{equation}
N \mathcal I(\gamma;{\rm F}^N_{\matr P_\beta \matr H})\label{linear}.
\end{equation}
The mutual information \eqref{linear} scales approximately linearly with the number of receive antennas, if it is close to
\begin{equation}
\beta N \mathcal I(\gamma;{\rm F}^N_{ \matr H}).
\end{equation}
Thus, in the high SNR limit, the deviation from the linear growth normalized to $N$ (the deviation from linear growth for short) is given by
\begin{align}
 \Delta \mathcal L(\beta; {\rm F}^N_\matr H)&\triangleq\lim_{\gamma\to \infty}\mathcal I(\gamma;{\rm F}^N_{\matr P_\beta\matr H})-\beta\mathcal I(\gamma;{\rm F}^N_{\matr H})\label{nonlin}\\&=\mathcal I_{0}(\gamma;{\rm F}^N_{\matr P_\beta\matr H})-\beta\mathcal I_{0}(\gamma;{\rm F}^N_{\matr H})\label{nonlinear} 
\end{align}
where $\matr H$ is assumed to have full rank almost surely. The full-rank assumption implies $\alpha^T_{\matr H}=\alpha^T_{\matr P_{\beta}\matr H}$ which is necessary in the definition \eqref{nonlin}. Otherwise, \eqref{nonlin} is divergent.
\begin{example}
Let $\matr H$ be unitary. Then, we have 
\begin{equation}
\Delta \mathcal L(\beta; {\rm F}^N_\matr H)=0.
\end{equation}
\end{example}
 
\subsection{The large-system limit consideration}
The deviation from linear growth \eqref{nonlinear} differs from the quantity $\chi_{\matr H}^T$ defined in \eqref{rateloss} only by the factor $\beta$ scaling the second term. Unlike $\chi_{\matr H}^T$, $\Delta \mathcal L$ does depend on the singular values of channel matrix. This makes the analysis somehow intractable. On the other hand, it is well-known that asymptotic results when the numbers of antennas grow large provide very good approximations already for systems with a dozen (or even less) of antennas in practice. Thus, we can resort to the asymptotic regime in the number of antennas to study the deviation from linear growth. To that end, in this section we make use of the following underlying assumption:
\begin{assumption}
The channel matrix $\matr H$ has full rank almost surely. Furthermore, $\matr H\matr H^\dagger$ is unitarily invariant, has a uniformly bounded spectral norm, and its empirical eigenvalue
distribution converges almost surely as $N\to \infty$. Moreover, $\Delta {\mathcal I}(1; {\rm F}_{\matr H})$ is finite.
\end{assumption}
We carry out the analysis on the basis of the LED function ${\rm F}_{\matr P_{\beta}\matr H}$. Specifically, we consider
\begin{equation}
\Delta \mathcal L(\beta; {\rm F}_\matr H)=\mathcal I_{0}(\gamma;{\rm F}_{\matr P_\beta\matr H})-\beta\mathcal I_{0}(\gamma;{\rm F}_{\matr H}).
\end{equation}
When we interpret the asymptotic results in the numerical investigations we assume that
\begin{equation}
\lim_{N\to \infty} {\rm E}[\mathcal I_0(\gamma;{\rm F}^N_{\matr P_{\beta}\matr H})]=\mathcal I_0(\gamma;{\rm F}_{\matr P_{\beta}\matr H})\ ,\quad \beta\leq 1. \label{key}
\end{equation}
It is easy to show that the convergence \eqref{key} is a mild assumption for $\beta<1$: as ${\rm F}_{\matr H}$ is assumed to have a compact support, ${\rm F}_{\matr P_{\beta}\matr H}$ has a compact support too, see \cite[Corollary~1.14]{Nica0}. Note that a compactly supported probability distribution can be uniquely characterized by its moments. This fact allow us to use the machinery provided in Proposition~1 in Appendix~C. Specifically, $\sup_{N}{\rm E}[\int  x^{-1} {\rm d}\tilde{\rm F}^N_{\matr P_{\beta}\matr H}(x)]<\infty$ is sufficient for \eqref{key} to hold. Indeed this is a reasonable condition for $\beta<1$ since 
\begin{equation}
\int \frac 1 x \;{\rm d}\tilde{\rm F}_{\matr P_{\beta}\matr H}(x)\ ,\quad  0<\beta<1 \label{harmean}
\end{equation} 
is strictly increasing with $\beta$, see Remark~\ref{defh}.

\begin{example}\label{correction}
Let the entries of $\matr H$ be iid with zero mean and variance $\sigma^2/N$. Then, we have 
\begin{equation}
\Delta \mathcal L(\beta; {\rm F}_\matr H)=(\beta-1)\log_2(1-\beta) \label{iid}
\end{equation}
where by convention $0\log_20=0$. 
\proof{See Appendix \ref{Examp1}}.
\end{example} 
In other words, at high SNR the normalized mutual information of a MIMO system of sufficiently large dimensions with zero-mean iid channel entries grows approximately linearly with the minimum of the numbers of transmit and receive antennas up to 1st order and the deviation from the linear growth is close to $(\beta-1)\log_2(1-\beta)$. Figure~\ref{fig3} illustrates this behavior.
\begin{figure}
\epsfig{file=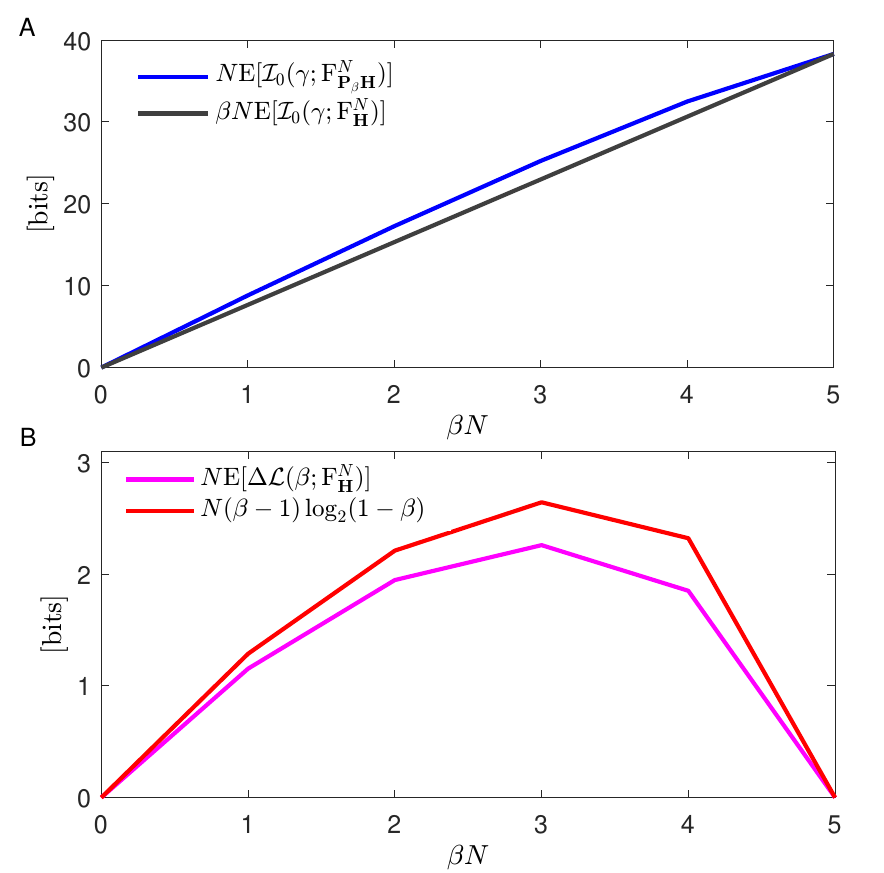,width=\columnwidth}
\caption{Ergodic multiplexing rate and corresponding linear growth (A) and (ergodic) deviation from linear growth (B) versus number of receive antennas $\beta N$. The entries of $\matr H\in \CC^{5\times 5}$ are iid Gaussian with zero mean and variance $1/5$.  The SNR is $\gamma=20~dB$.}\label{fig3}
\end{figure}

\subsection{The S-transform formulation}
The result in Example~\ref{correction} can be obtained from previous capacity results, e.g. \cite[Eq. (2.63)]{tulino}. We obtained it as a special case of the following lemma.
\begin{lemma}\label{deviation}Let $\matr H$ fulfill Assumption~1. Then, we have    
\begin{eqnarray}
\Delta \mathcal L(\beta; {\rm F}_\matr {H})= -\beta\int _{0}^{1}\log_2\frac{{\rm S}_{\matr H}(-\beta z)}{{\rm S}_{\matr H}(-z)}\;{\rm d}z.
\end{eqnarray}
\proof{See Appendix~\ref{Lemma4}.}
\end{lemma}
Alternatively, we may bypass the need for using the S-transform by invoking the following result:
\begin{remark}\label{defh}
Let $\matr H$ fulfill Assumption~1. Furthermore, let $\matr P_{t}$ be an $N$-dimensional projector with $0<t<1$.
Then, we have 
\begin{equation}
{\rm S}_{\matr H}(-t)=\int \frac{1}{x}\;{\rm d\tilde F}_{\matr P_{t}\matr H}(x)\ , \quad 0<t<1. \label{mrn}
\end{equation}
\proof{See Appendix~\ref{Lemma4}.}
\end{remark}
The result in \eqref{mrn} also provides a convenient means to calculate the deviation from linear growth in the large-system limit. The right-hand side of \eqref{mrn} is nothing but the asymptotic inverse spectral mean of the channel matrix $\matr P_{t}\matr H$.

\subsection{The universality related to the SNR range}
Note that the difference $\mathcal I(\gamma; {\rm F}_{\matr P_\beta\matr H})-\beta\mathcal I(\gamma; {\rm F}_{\matr H})$ converges to $\Delta \mathcal L(\beta; {\rm F}_\matr {H})$ as the SNR tends to infinity, see \eqref{nonlinear}. In Appendix~\ref{devsupproof} we show that this difference actually increases with SNR unless ${\rm F}_{\matr H}$ is a Dirac distribution function. Thus, we have the following universal characterization over the whole SNR range.
\begin{remark} \label{devsup}
Let $\matr H$ fulfill Assumption~1. Then, we have 
\begin{align}
\Delta \mathcal L(\beta,{\rm F}_{\matr H})= \sup_{\gamma}\left\{\mathcal I(\gamma; {\rm F}_{\matr P_\beta\matr H})-\beta\mathcal I(\gamma; {\rm F}_{\matr H})\right\}. \label{devsup1}
\end{align}
\proof{See Appendix~\ref{devsupproof}.}
\end{remark}
\subsection{The additive property}
We now draw the attention to another important property of the deviation from linear growth:
\begin{theorem}\label{add} 
	Let $\matr X$ and $\matr Y$ be independent $\CC^{N\times N}$ random matrices. Moreover, let $\matr X$ and $\matr Y$ fulfill Assumption~1.
	Then, we have
	\begin{equation}
	\Delta \mathcal L(\beta; {\rm F}_\matr {XY})=\Delta \mathcal L(\beta; {\rm F}_\matr {X})+\Delta \mathcal L(\beta; {\rm F}_\matr {Y}).
	\end{equation}
	\proof{See Appendix~\ref{The3}.}  
\end{theorem}
\begin{example}
	Consider a random matrix defined as
	\begin{equation}
	\matr H=\prod_{m=1}^{M}\matr A_m  \label{product}
	\end{equation}
	where the $N\times N$ matrices $\matr A_m$,  $m=1,\dots, M$, are independent, have iid entries with zero mean and variance $\sigma^2/N$. Then we have almost surely
	\begin{align}
	\Delta \mathcal L(\beta;{\rm F}_\matr {H})&=M\Delta \mathcal L(\beta; {\rm F}_{\matr {A}_1}) \\ 
	&= M(\beta-1)\log_2(1-\beta).
	\end{align}
\end{example}
As mentioned previously, the crosstalk due to non-orthogonal columns in $\matr H$ affects the mutual information in a way that is non-linear in the number of antennas. Thus, it causes the deviation from linear growth. Let us be more precise here and (inspired from \cite[Eq. (1)]{Litwin}) introduce the concept of \emph{crosstalk ratio}:
\begin{equation}
{\rm CT}_{\matr H}\triangleq \lim_{N\to \infty} \frac{\sum_{i=1}^{N}\sum_{j<i}{\vert\matr h_{i}^\dagger \matr h_{j}\vert^2}}{\sum_{i=1}^{N}\vert\matr h_{i}^\dagger\matr h_{i}\vert^2}.\label{CTR}
\end{equation}
Here $\matr h_{i}$ denotes the $i$th column of $\matr H$. For example, for an unitary matrix $\matr H$, we have ${\rm CT}_{\matr H}=0$. As a second example, let the entries of $\matr H$ be iid complex Gaussian with zero mean and variance $1/N$. Then, from \eqref{ACTR} we get
\begin{equation}
{\rm CT}_{\matr H}=\frac{1}{2}.
\end{equation}
We next show that the crosstalk ratio has the same additive property as the deviation from linear growth.
\begin{remark}\label{CTadd}
Let $\matr X$ and $\matr Y$ be independent $\CC^{N\times N}$ random matrices. Moreover let $\matr X$ and $\matr Y$ fulfill Assumption~1. Then we have
\begin{equation}
		{\rm CT}_{\matr X\matr Y}={\rm CT}_{\matr X}+{\rm CT}_{\matr Y}.
\end{equation}
		\proof{See Appendix~\ref{R7}.}  
\end{remark}
%%%%%%%%%%%%
\section{Conclusions}\label{conc}

A variation of the number of antennas in a MIMO system affects the mutual information at asymptotically large SNR in following way: If the minimum number of antennas at transmitter and receiver side stays unaltered, the change of mutual information depends only on the system dimensions and the matrix of left (or right) singular vectors of initial channel matrix but {\em not on its singular values}. For channel matrices that are unitarily invariant from left (or right) this change of mutual information in the ergodic sense can be expressed with a simple analytic function of the system dimensions. Moreover, the large system limit of this expression involves only the binary entropy functions of the aspect ratios of two varying channel matrices -- the one before and the one after altering the number of antennas.

Mutual information grows only approximately linear with the minimum of the system dimensions even at high SNR. 
This deviation from that linear growth, i.e.\ the error of the linear approximation, does depend on the singular values of the channel matrix. It can be quantified and has the following remarkable property in the large system limit:
For certain factorizable MIMO channel matrices, the deviation is the sum of the deviations of the individual factors.

The results derived in this work for asymptoticly large SNR are least upper bounds over the whole SNR range. This gives them a universal character.

Finally, a fundamental relation between mutual information and its affine approximation (the multiplexing
rate) was unveiled. This relation can be conveniently described via the S-transform of free
probability.

%\section*{Acknowledgment}
%%This work was supported by the European Commission in the framework of the FP7 Network of Excellence in Wireless COMmunications NEWCOM$\sharp$ (Grant agreement no. 318306). 
%This work was partially supported by the research project VIRTUOSO funded by Intel Mobile Communications, Keysight, Telenor, Aalborg University, and the Danish National Advanced Technology Foundation
\appendices

%%%%%%%%%%%%%%%%%%%%%%%%%%%%%
\section{Preliminaries}
\begin{lemma} \label{entropy} 
Let $p \in [0,1]$. Then, we have
\begin{eqnarray}
\int\limits_{0}^{p}\log_2\frac{1-z}{p-z}\;{\rm d}z=H(p). \label{BEF}
\end{eqnarray}
\end{lemma}
\begin{proof}
We first recast \eqref{BEF} into the equivalent identity
\begin{eqnarray}
\lim_{x\to p}\int\limits_{0}^{x}\log_2\frac{1-z}{x-z}\;{\rm d}z=H(p). \label{57}
\end{eqnarray}
To prove \eqref{57}, we first apply a variable substitution
\begin{equation}
\int\limits_{0}^{x}\log_2\frac{1-t}{x-t}\;{\rm d}t=x\int\limits_{0}^{1}\log_2\frac{x^{-1}-z}{1-z}\;{\rm d}z \label{sub}
\end{equation}
and decompose the right hand side of (\ref{sub}) as 
\begin{equation}
x\int\limits_{0}^{1}\log_2(x^{-1}-z)\;{\rm d}z -x \int\limits_{0}^{1}\log_2{(1-z)}\;{\rm d}z.
\end{equation}
Define $u\triangleq\log_2(x^{-1}-z)$ and $v=z$. Applying the integration by part rule, we obtain for the first integral:
\begin{eqnarray}
\int\limits_{0}^{1}\log_2\left(x^{-1}-z\right)\; {\rm d}z&=&\left. uv\right |_{0}^{1}-\int\limits_{0}^{1} v\;{\rm d}u\\
&=& x^{-1}H(x)-\log_2 e. \label{integral}
\end{eqnarray}
Using (\ref{integral}), we compute the second integral:
\begin{eqnarray}
\int\limits_{0}^{1}\log_2{(1-z)}\;{\rm d}z=\lim_{x\to 1}\int\limits_{0}^{1}\log_2(x^{-1}-z)\;{\rm d}z=-\log_2 e.
\end{eqnarray}
This completes the proof. 
\end{proof}

\begin{lemma}\label{suminv}\cite{Miller} Let $\matr A$ and $\matr A+\matr B$ be invertible and $\matr B$ have rank $1$. Furthermore let $g\triangleq{\rm tr}(\matr B\matr A^{-1})\neq -1$. Then, we have 
\begin{equation}
\left(\matr A+\matr B\right)^{-1}= \matr A^{-1}- \frac{1}{g+1}\matr A^{-1}\matr B \matr A^{-1}.
\end{equation}
\end{lemma}
\begin{lemma}\label{Smap}\cite[Lemma~2 \& Lemma~4]{hager} 
Let ${\rm F}$ be a probability distribution function with support in $[0,\infty)$ and $\rm S$ its S-transform. Moreover, let ${\rm F}$ be not a Dirac distribution function. Then, $\rm S$ is strictly decreasing on $(-\alpha,0)$ with $\alpha\triangleq 1-{\rm F}(0)$. In particular, we have 
\begin{align}
\lim_{z\to 0^-}{\rm S}(z)& = \left(\int x \; {\rm dF}(x)\right)^{-1} \\
\lim_{z\to -\alpha^+}{\rm S}(z)& = \int \frac{1}{x}\; {\rm dF}(x) \label{sinv}
\end{align} 
where we use the convention $1/0=\infty$ in \eqref{sinv} when ${\rm F}(0)>0$. 
\end{lemma}

\begin{theorem}\label{moller}\cite[Proposition 1]{hager} Let ${\rm F}$ be a probability distribution function with support in $(0,\infty)$ and $\rm S$ its S-transform. Then $\int\vert\log x\vert\; {\rm d F}(x)$ is finite if, and only if, $\int_{0}^{1} \vert\log  {\rm S}(-z)\vert\; {\rm d}z$ is finite. If either of these integrals is finite, 
\begin{equation}
\int \log (x)\;{\rm d F}(x)=-\int\limits_{0}^{1}\log {\rm S}(-z)\;{\rm d}z.
\end{equation}
\end{theorem}
\begin{theorem}\label{convergeceS}
For $n \in \mathbb N^{+}\triangleq\{1,2,...\}$ let ${\rm F}^n$ be probability distribution functions on $[0,\infty)$. Furthermore let $1-{\rm F}^n(0)=\alpha>0$, $\forall n\in \mathbb N^{+}$. Moreover let ${\rm S}^n$ denote the S-transform of ${\rm F}^n$. Then if ${\rm F}^n$ converges weakly to a probability distribution function ${\rm F}$ as $n \to \infty$, we have 
\begin{equation}
\lim_{n\to \infty}{\rm S}^n(z)={\rm S}(z), \quad -\alpha< z<0
\end{equation}
where $\rm S$ is the S-transform of $\rm F$.
\begin{proof}
Let us consider the function, (see \eqref{psi})
\begin{equation}
\Psi^n(z)\triangleq \int \frac{zx}{1-zx}\;{\rm dF}^n(x),\quad  -\infty<z<0. \label{Psi}
\end{equation}
For $z \in (-\infty,0)$, $z\to \frac{zx}{1-zx}$ is bounded and continuous. Hence, the weak convergence of ${\rm F}^n$ implies that 
\begin{equation}
\lim_{n\to \infty} \Psi^n(z)=\Psi(z), \quad  -\infty<z<0.
\end{equation}
Furthermore, $\Psi^n(z)$ is a strictly increasing homeomorphism of $(-\infty,0)$ onto $(-\alpha, 0)$ \cite{hager}. This implies that (see e.g. \cite[Proposition 0.1]{resnick})
\begin{equation}
\lim_{n\to \infty} (\Psi^{n})^{-1}(z)=\Psi^{-1}(z), \quad -\alpha<z<0.
\end{equation}
This completes the proof.
\end{proof}
\end{theorem} 

\begin{lemma}\label{projection}Consider a random matrix $\matr X$ and a projector $\matr P_{\beta}$. Assume that $\matr X^\dagger\matr X$ and  $\matr P_\beta^\dagger \matr P_\beta$ are asymptotically free. Then, we have  
\begin{equation}
{\rm S}_{\matr X\matr P_{\beta}^\dagger}(z)={\rm S}_{\matr X}(\beta z).
\end{equation}
\end{lemma}

\begin{proof}
The S-transform of $\matr P_\beta$ reads \cite[Example 2.32]{tulino} 
\begin{eqnarray}
{\rm S}_{\matr P_{\beta}}(z)=\frac{z+1}{z+\beta}. \label{proj}
\end{eqnarray}
By invoking the identity \cite[Theorem 2.32]{tulino} and the asymptotic freeness between $\matr X^\dagger \matr X$ and $\matr P_\beta^\dagger \matr P_\beta$, we obtain 
\begin{eqnarray}
{\rm S}_{\matr X \matr P_{\beta}^\dagger}(z)&=&\frac{z+1}{z+1/\beta}{\rm S}_{\matr P_\beta}(\beta z){\rm S}_{\matr X}(\beta z) \\ \label{sams}
&=& {\rm S}_{\matr X}(\beta z).\label{same}
\end{eqnarray}
\end{proof}

\begin{remark}\label{Jacobiremark} 
Let $\matr H = \matr P_{\beta_2} \matr U \matr P_{\beta_1}^{\dagger}$ 
with $\matr U$ an $N$-dimensional Haar unitary. Then, we have almost surely
\begin{equation}
	{\rm S}_{\matr H}(z)=\frac{1+\beta_1z}{\beta_2+\beta_1z}.\label{jacobiremak}
\end{equation}
\begin{proof}
By invoking Lemma~\ref{projection} to \eqref{proj} we obtain \eqref{jacobiremak}.
\end{proof}
\end{remark}

%%%%%%%%%%%%%%%%%%%%%%%%%
\section{Proof of Theorem~\ref{main}}\label{Lemma4new}
\subsection{Proof of \eqref{etaalpha}}
By definition, $\mathcal I(\gamma;{\rm F}_{\matr H}^T)<\infty$. Then, from identity \cite[Eq.~(5)]{hager} we write
\begin{equation}
\mathcal{I}(\gamma;{\rm F}_{\matr H}^T)=-\int_{0}^1\log_2 (s)\partial\Psi_{\sqrt{\gamma}\matr H}^{T}(-s)\;{\rm d}s \label{start} 
\end{equation}
where $\partial\Psi_{\matr H}^{T}(\omega)\triangleq\left.\frac{{\rm d}\Psi_{\matr H}^T(x)}{{\rm d}x}\right\vert_{x=\omega}$. 

At this stage we point out two identities:
\begin{align}
\Psi_{\sqrt{\gamma}\matr H}^T(-1)+1&= \Psi_{\matr H}^T(-\gamma)+1=\eta_{\matr H}^T(\gamma) \label{uplim}
\\ \lim_{x\to 0^{-}}\Psi_{\sqrt{\gamma}\matr H}^T(x)+1&= 1.\label{lowlim}
\end{align} 
Now we apply the variable substitution $z\triangleq \Psi_{\sqrt\gamma\matr H}^T(-s)+1$ in the integral in \eqref{start}. Notice that with this substitution the upper and lower limits of this integral read \eqref{uplim} and \eqref{lowlim}, respectively. As a result \eqref{start} is recast in the form
\begin{equation}
\mathcal{I}(\gamma;{\rm F}_{\matr H}^T)=\int_{1}^{\eta_{\matr H}^T}\log_2\left(-\Psi_{\sqrt{\gamma}\matr H}^{T^{<-1>}}(z-1)\right)\;{\rm d}z\\
\end{equation}
with $\Psi_{\matr H}^{T^{<-1>}}$ denoting the inverse of $\Psi_{\matr H}^{T}$. Then, by the definition of the S-transform, see \eqref{defs}, we obtain
\begin{align}
\hspace{-0.1cm}\mathcal{I}(\gamma;{\rm F}_{\matr H}^T)&=\int^{\eta_{\matr H}^T}_1\log_2 \frac{1-z}{z} {\rm d}z  +\int^{\eta_{\matr H}^T}_1\log_2 {\rm S}_{\sqrt{\gamma}\matr H}^T(z-1)\;{\rm d}z  \\
&=H(\eta_{\matr H}^T)+\int^{\eta_{\matr H}^T}_1\log_2 {\rm S}_{\sqrt{\gamma}\matr H}^T(z-1)\;{\rm d}z  \\
&=H(\eta_{\matr H}^T)-\int_{0}^{1-\eta_{\matr H}^T}\log_2{\rm S}_{\sqrt{\gamma}\matr H}^T(-z)\;{\rm d}z. \label{lemma1f}
\end{align}
Finally, we obtain \eqref{etaalpha} by using the scaling property of the S-transform \cite[Lemma 4.2]{debbah}.
\subsection{Proof of \eqref{mult}}
Let $\tilde {\rm S}^T_{\matr H}$ be the S-transform of $\tilde {\rm F}^T_{\matr H}$. By using \cite[Theorem 2.32]{tulino} we write
\begin{equation}
\tilde {\rm S}^T_{\matr H}(z)=\frac{z+1}{z+1/\alpha^T_\matr H}{\rm S}^T_{\matr H}(\alpha^T_\matr H z), \quad -1<z<0.
\end{equation}
Note that $\tilde {\rm F}^T_{\matr H}$ is an empirical distribution function. Thus, $\log_2(x)$ is absolutely integrable over it. We use Theorem~\ref{moller} and Lemma~\ref{entropy} to complete the proof:
\begin{align}
\mathcal I_{0}&(\gamma;{\rm F}^T_{\matr H})=\alpha^T_\matr H\log_2\gamma-\alpha^T_\matr H\int\limits_{0}^{1}\log_2\tilde {\rm S}^T_{\matr H}(-x)\; {\rm d}x   \\
&=\alpha^T_\matr H\log_2\gamma -\alpha^T_\matr H\int\limits_{0}^{1}\log_2\frac{1-x}{1/\alpha^T_\matr H-x}{\rm S}^T_{\matr H}(-\alpha^T_\matr H x)\;{\rm d}x \\
&=\alpha^T_\matr H\log_2\gamma-\int\limits_{0}^{\alpha^T_\matr H}
\log_2\frac{\alpha_\matr H^T-x}{1-x}{\rm S}^T_{ \matr H}(- x)\; {\rm d}x  \\
&=\alpha^T_\matr H\log_2\gamma+H(\alpha^T_\matr H)-\int\limits_{0}^{\alpha^T_\matr H}\log_2{\rm S}^T_{\matr H}(-x)\; {\rm d}x.
\end{align}

%%%%%%%%%%%%%%%%
\section{On the convergence of Mutual information and Multiplexing Rate}
In this section we provide some sufficient conditions that guarantee the convergence of the mutual information \eqref{mut} and multiplexing rate (see \eqref{dev}) in the large system limit.
\begin{proposition}\label{convergence}
	As $R,T\to \infty$ with the ratio $\phi\triangleq T/R$ fixed let $\matr H^\dagger \matr H$ have a LED ${\rm F}_{\matr H}$. Furthermore, let 
	\begin{equation}
	\sup_{T}\int x \; {\rm d F}^T_{\matr H}(x)<\infty \quad \text{a.s.}. \label{conI}
	\end{equation}
	Then we have almost surely
	\begin{eqnarray}
	\lim_{T\to \infty}\mathcal I(\gamma;{\rm F}^T_\matr H)=\mathcal I(\gamma;{\rm F}_\matr H).
	\end{eqnarray}
	Moreover if in addition
	\begin{equation}
	\sup_{T}\int \frac 1 x \; {\rm d}\tilde{\rm F}^T_{\matr H}(x)<\infty \quad \text{a.s.}  \label{inmean}
	\end{equation}
	we have almost surely
	\begin{equation}
	\lim_{T\to \infty} \mathcal I_0(\gamma;{\rm F}^T_{\matr H})=\mathcal I_0(\gamma;{\rm F}_{\matr H}). \label{careful1}
	\end{equation}  
\end{proposition}

Condition~\eqref{conI} is reasonable in practice. Otherwise the power amplification per dimension of the MIMO system explodes as its dimensions grow to infinity. One can show that for rectangular and unitarily invariant channel matrices, the condition \eqref{inmean} is reasonable too due to the strict decreasing property of the function of $\beta$ in \eqref{harmean}. However, it might not hold when the channel matrix is square. As an example, consider a channel matrix $\matr H$ whose entries are iid with zero mean and variance $\sigma^2/T$. Then, condition~\eqref{inmean} holds if $\phi\neq 1$, but is violated if $\phi = 1$. Indeed the latter case turns out critical for the ``$\log \det$" convergence of the zero-mean iid matrix ensemble, e.g. see \cite{Dag,tulino}. Nevertheless, both \cite[Proposition~2.2]{Tao} and numerical evidence lead us to conjecture that \eqref{careful1} holds when $\phi=1$ as well. Thus, we conclude that the asymptotic convergence of the multiplexing rate, i.e. \eqref{careful1}, is a mild assumption in practice.
\subsection*{Proof of Proposition~1}
For the sake of readability of the proof, whenever we use the limit operator indicating that  $T$ tends to infinity, we implicitly assume that the ratio $\phi=T/R$ is fixed.    

For convenience we define 
\begin{equation}
\matr Y\triangleq\matr {\bf I}+\gamma \matr H^\dagger \matr H.
\end{equation}
By Theorem~\ref{moller} we have 
\begin{equation}
\mathcal I(\gamma;{\rm F}^T_\matr H)=-\int\limits_{-1}^{0}\log_2{\rm S}^T_{\sqrt{\matr Y}}(z)\; {\rm d}z.
\end{equation}
The function ${\rm S}^T_{\sqrt{\matr Y}}$ is strictly decreasing on $(-1,0)$ if, and only if, ${\rm F}^T_{\matr H}$ is not a Dirac distribution function, see Lemma~\ref{Smap}. If ${\rm F}_{\matr H}^T$ is a Dirac distribution function then ${\rm S}^T_{\sqrt{\matr Y}}$ is a constant function. Without loss of generality, we can assume that ${\rm F}^T_{\matr H}$ is not a Dirac distribution function. Then, by invoking Lemma~\ref{Smap} again we have
\begin{equation}
\left(\frac{1}{T}{\rm tr}(\matr Y)\right)^{-1}<{\rm S}^T_{\sqrt{\matr Y}}(z)<\frac{1}{T}{\rm tr}(\matr Y^{-1}), \quad-1<z<0. \label{less}
\end{equation}
For convenience we define the random variable 
\begin{equation}
M\triangleq\sup_{T}\int x \; {\rm d F}^T_{\matr H}(x)\quad \text{s.t. }  \phi=\frac{T}{R} \label{ubound}.
\end{equation}
Since the upper bound in \eqref{less} is smaller than one we have 
\begin{align}
\vert\log_2 {\rm S}^T_{\sqrt{\matr Y}}(z)\vert&=-\log_2 {\rm S}^T_{\sqrt{\matr Y}}(z),\quad-1<z<0\\
&<\log_2\frac{1}{T}{\rm tr}(\matr Y)\\
&\leq\log_2(1+\gamma M) \label{enough}.
\end{align}
Because of \eqref{enough}, we can apply Lebesgue's dominated convergence theorem \cite[Theorem~10.21]{browder}:
\begin{eqnarray}
\lim_{T\to \infty}\mathcal I(\gamma;{\rm F}^T_\matr H)=-\int\limits_{-1}^{0}\lim_{T\to \infty}\log_2{\rm S}^T_{\sqrt{\matr Y}}(z)\; {\rm d}z. \label{slim}
\end{eqnarray}
By invoking Theorem~\ref{convergeceS} we complete the proof of \eqref{conI}:
\begin{align}
\lim_{T\to \infty}\log_2{\rm S}^T_{\sqrt{\matr Y}}(z)&=\log_2\lim_{T\to \infty}{\rm S}^T_{\sqrt{\matr Y}}(z)\\ &
=\log_2{\rm S}_{\sqrt{\matr Y}}(z)\label{Slimit}.
\end{align}

To prove \eqref{careful1}, we use the same arguments for $\mathcal I_0(\gamma,{\rm F}^T_\matr H)$ as for $\mathcal I(\gamma,{\rm F}^T_\matr H)$. In particular, by invoking Lemma~\ref{Smap} again we can write
\begin{equation}
\left(\int x\; {\rm d}{\rm \tilde F}^T_{\matr H}(x)\right)^{-1} <{\rm \tilde S}^T_{\matr H}(z)<\int \frac{1}{x} \; {\rm d}{\rm \tilde F}_{\matr H}^{T}(x),\quad  -1<z<0   \label{less1}
\end{equation}
with ${\rm \tilde S}^T_{\matr H}$ denoting the S-transform of ${\rm \tilde F}_{\matr H}^T$. Unlike \eqref{less}, the right-most integral is not bounded in general, so we need the additional assumption \eqref{inmean}. This completes the proof.

%%%%%%%%%%%%%%%%%%%%%%%%%
\section{Proof of Example~\ref{Example1}}\label{PExample1}
With a convenient re-parameterization of \cite[Eq. (19)]{ralfa} we write
\begin{equation}
{\rm S}_{\matr H}(z)=\prod_{n=1}^{N}\frac{\rho_n}{z+\rho_n}. \label{Siid}
\end{equation}
From Theorem~\ref{main} we have 
\begin{align}
\mathcal I(\gamma;{\rm F}_{\matr H}) &=H(\eta_{\matr H})+(1-\eta_{\matr H})\log_2\gamma \nonumber \\ &+\sum_{n=1}^N\int_{0}^{1-\eta_{\matr H}}\log_2(1-\frac{z}{\rho_n})\;{\rm d}z\label{int}.
\end{align}
We can write the integral terms in \eqref{int} as
\begin{align}
&\int_{0}^{1-\eta_{\matr H}}\log_2(1-\frac {z}{\rho_n})\;{\rm d}z=\nonumber \\& \log_2\frac{1-\eta_{\matr H}}{\rho_n} +\int_{0}^{1}\log_2(\frac{\rho_n}{1-\eta_{\matr H}}-z)\;{\rm d}z \label{int2}
\end{align}
for $n\in [1,N]$. By invoking the result in \eqref{integral} we obtain \eqref{iidmut}. 

From the linearity property of the Lebesgue integral, it is easy to show that $\int_{0}^{1} \vert \log_2 {\rm \tilde S}_{\matr H}(-z){\rm d}z\vert$ is finite, which implies that $\int \vert \log(x)\vert {\rm d\tilde F}_{\matr H}(x)$ is finite too, due to Theorem~\ref{moller}. Thus, the multiplexing rate is obtained by replacing the term $(1- \eta_{\matr H})$ in \eqref{iidmut} with $\alpha_{\matr H}$ (due to Theorem~\ref{main}). This leads to \eqref{iidmultip}. Finally, we note that if $\alpha_{\matr H}<1$ the S-transform ${\rm S}_{\matr H}(z)$ diverges as $z\to (-\alpha_{\matr H})$, see Lemma~\ref{Smap}. Thus, from \eqref{Siid} the unique solution of $\alpha_{\matr H}$ is $\alpha_{\matr H}=\min(1, \rho_1,\rho_2,\dots,\rho_N)$.

%%%%%%%%%%%%%%%%%%%%%%%%%
\section{Proof of Example~\ref{Example2}}\label{PExample2}
Recall \eqref{jacobiremak}:
\begin{equation}
{\rm S}_{\matr H}(z)=\frac{1+\beta_1z}{\beta_2+\beta_1z}. \label{Sjac}
\end{equation}
Moreover, notice that $\alpha_{\matr H}=1-{\rm F}_{\matr H}(0)=\min(1, \beta_2/\beta_1)$. For convenience let $a\triangleq 1-\eta_{\matr H}(\gamma)<\alpha_{\matr H}$. Then, we have
\begin{align}
&\int_{0}^{a} \log_2{\rm S}_{\matr H}(-z)\;{\rm d}z= a\int_{0}^{1}\log_2\frac{1-\beta_1at}{\beta_2-\beta_1at}\;{\rm d}t \\
&= \frac{H(\beta_1 a)}{\beta_1}-\frac{\beta_2}{\beta_1}H\left( \frac{\beta_1}{\beta_2}a\right)\label{116}
\end{align}
where the result \eqref{116} follows from the identity (\ref{integral}). We obtain \eqref{jacobimut} from \eqref{etaalpha} with \eqref{116} inserted in \eqref{goodform}. Moreover, by the definition of the S-transform we have
\begin{equation}
\beta_1(1-z)\Psi_{\matr H}^2(z)+(1-(\beta_1+\beta_2)z)\Psi_{\matr H}(z)-\beta_2z=0.\label{quadratic}
\end{equation}
Note that $1+\Psi_{\matr H}(-\gamma)=\eta_{\matr H}(\gamma)$. Thus, \eqref{quadratic} has two solutions for $\eta_{\matr H}(\gamma)$. Only one fulfills the properties of $\eta_{\matr H}(\gamma)$ in \cite[pp.~41]{tulino}. Specifically, from the property $\eta_{\matr H}(\gamma)\to 1$ as $\gamma\to 0$ we conclude that  \eqref{etajacobi} is this solution. Finally it is also easy to show that $\int_{0}^{1} \vert \log_2 {\rm \tilde S}_{\matr H}(-z)\vert\;{\rm d}z$ is finite in this case. This implies that $\int \vert \log(x)\vert {\rm d\tilde F}_{\matr H}(x)$ is finite too. Thus, the multiplexing rate is obtained by replacing the term $(1- \eta_{\matr H})$ in \eqref{jacobimut} with $\alpha_{\matr H}$, which leads to \eqref{jocobimultip}. 
%%%%%%%%%%%%%%%%%%%%%%%%%

%%%%%%%%%%%%%%%%%%%%%%%%%
\section{Proof of Remark~\ref{SupI}}\label{PSupI}
We first point out the relationship \cite{tulino}
\begin{align}
\frac{{\rm d}\{\mathcal I(\gamma;{\rm F}^T_{\matr H})-\mathcal I(\gamma;{\rm F}^T_{\matr P_{\beta}\matr H})\}}{{\rm d}\gamma}=\frac{\eta^T_{\matr P_\beta \matr H}(\gamma)-\eta^T_{\matr H}(\gamma)}{\gamma \ln 2}. 
\end{align}
Hence, to prove Remark~\ref{SupI} we simply need to show that
\begin{equation}
{\rm tr}\left\{(\matr {\rm I}+\gamma \matr H^\dagger \matr P_\beta^\dagger \matr P_{\beta}\matr H)^{-1}-(\matr {\rm I}+\gamma \matr H^\dagger \matr H)^{-1}\right\} \geq 0 \label{supprof}
\end{equation}
where the equality holds when $\beta=1$. To prove \eqref{supprof} it is sufficient to consider the removal of a single receive antenna, i.e. 
$\beta=(R-1)/R$. It is immediate that
\begin{equation}
\matr H^\dagger\matr H=\matr H^\dagger \matr P_\beta^\dagger \matr P_{\beta}\matr H+ \matr h_{R}^\dagger\matr h_R  
\end{equation}
with $\matr h_R\in \CC^{1\times T}$ representing the $R$th row of $\matr H$. Then \eqref{supprof} follows directly from Lemma~\ref{suminv} in Appendix~A.

%%%%%%%%%%%%%%%%%%%%%%%%%
\section{Proof of Remark~\ref{Cap}}\label{PCap}
We decompose the capacity expression in \eqref{ccap} as
\begin{equation}
{\mathcal C}(\gamma, {\rm F}_{\matr P_{\beta}\matr H}^{T})= \mathcal I_0(\gamma;{\rm F}_{\matr P_{\beta}\matr H\sqrt{\matr Q^*}}^T)+ \Delta\mathcal I(\gamma;{\rm F}_{\matr P_{\beta}\matr H\sqrt{\matr Q^*}}^T) 
\end{equation}
with $\matr Q^{\star}$ denoting the capacity achieving covariance matrix. We define 
\begin{equation}
 {\mathcal C}_0(\gamma, {\rm F}_{\matr P_{\beta}\matr H}^{T})\triangleq \operatorname*{\max}_{\substack{
 		\matr Q\geq 0\\
 		{\rm{tr}}(\matr Q)=T}}\mathcal I_0(\gamma;{\rm F}_{\matr P_{\beta}\matr H\sqrt{\matr Q}}^T). \label{C0}
\end{equation}
In particular, by the definitions in \eqref{ccap} and \eqref{C0} we have ${\mathcal C}_0(\gamma, {\rm F}_{\matr P_{\beta}\matr H}^{T})\geq \mathcal I_0(\gamma;{\rm F}_{\matr P_{\beta}\matr H\sqrt{\matr Q^*}}^T)$. Hence, we have
\begin{equation}
\lim_{\gamma \to \infty} \mathcal C_{0}(\gamma; {\rm F}_{\matr P_{\beta}\matr H}^T)-\mathcal C(\gamma; {\rm F}_{\matr P_{\beta}\matr H}^T)\geq 0 \label{capdif}.
\end{equation}

Since $\matr H^\dagger \matr P_\beta^\dagger \matr P_\beta \matr H$ has almost surely full rank, we have $\alpha^T_{\matr P_{\beta}\matr H\sqrt{\matr Q}}=\alpha^T_{\sqrt{\matr Q}}$ and thereby
\begin{equation}
\mathcal I_0(\gamma;{\rm F}_{\matr P_{\beta}\matr H\sqrt{\matr Q}}^T)= \alpha^T_{\sqrt{\matr Q}}\log_2 \gamma +  \alpha^T_{\sqrt{\matr Q}}\int \log_2 x~ {\rm d{\tilde F}}^T_{\matr P_{\beta}\matr H\sqrt{\matr Q}}(x). \label{maxI0}
\end{equation}
For a sufficiently large SNR a full-rank matrix $\matr Q$ maximizes \eqref{maxI0}. Therefore, to prove the result we can assume without loss of generality that $\matr Q$ has full rank. Doing so, we have 
\begin{align}
\mathcal I_0(\gamma;{\rm F}_{\matr P_{\beta}\matr H\sqrt{\matr Q}}^T)=&\log_2\gamma +\frac{1}{T}\log_2\det{\matr H^\dagger \matr P_\beta^\dagger \matr P_\beta \matr H}\nonumber \\&+\frac{1}{T}\log_2\det{\matr Q}.\label{opt}
\end{align}
Due to the constraint $\text{tr}(\matr Q)=T$, the identity operator maximizes (\ref{opt}). Hence, from \eqref{capdif} we have
\begin{equation}
 \lim_{\gamma \to \infty} \mathcal I_{0}(\gamma; {\rm F}_{\matr P_{\beta}\matr H}^T)-\mathcal C(\gamma; {\rm F}_{\matr P_{\beta}\matr H}^T)\geq 0 \label{capdiff}.
\end{equation}
On the other hand we have
\begin{equation}
\mathcal I_{0}(\gamma;{\rm F}_{\matr P_\beta\matr H}^T)< \mathcal I (\gamma;{\rm F}_{\matr P_\beta\matr H}^T)\leq \mathcal C(\gamma; {\rm F}_{\matr P_{\beta}\matr H}^T).
\end{equation}
Thus \eqref{capdiff} must be zero. This completes the proof.

%%%%%%%%%%%%%%%%%%%%%%%%%
\section{Proof of Theorem~\ref{Theorem2}}\label{PTheorem2}
We prove \eqref{general} which is a generalization of Theorem~\ref{Theorem2}. We make use of \eqref{gmodel} to write
\begin{equation}
\det \matr H^\dagger\matr P_\beta^\dagger \matr P_{\beta} \matr H = \det \matr \Lambda_{\rm T}^\dagger \matr \Lambda_{\rm T}\det \matr {\tilde H}^\dagger\matr \Theta_\beta \matr {\tilde H}
\end{equation}
where $\Theta_\beta\triangleq{{{\matr \Lambda_{\rm R}}}}^\dagger \matr P_{\beta}^\dagger\matr P_{\beta}{{{\matr \Lambda_{\rm R}}}}$ for $\beta\leq 1$. Hence, from \eqref{ratef} the rate loss reads as
\begin{equation}
\chi_{\matr H}^T(R,\beta R)=\frac{1}{T}\log_2 \frac{\det \matr {\tilde H}^\dagger\matr \Theta_1 \matr {\tilde H}}{\det \matr {\tilde H}^\dagger\matr \Theta_\beta \matr {\tilde H}}.\label{ratefg}
\end{equation} 
To simplify this expression, we consider the singular value decomposition of $\matr{\tilde H}$
\begin{equation}
\matr {\tilde H}=\matr{\tilde L} \matr [\matr \Sigma \vert \matr 0]^\dagger \matr {\tilde R} \label{singular} 
\end{equation}
where $\matr {\tilde L}$ and $\matr {\tilde R}$ are respectively $R\times R$ and $T\times T$ unitary matrices, $\matr \Sigma$ is a $T\times T$ positive diagonal matrix and $\matr 0$ is a $(R-T)\times T$ zero matrix. Remark that we can actually write \eqref{singular} as
\begin{align}
\matr {\tilde H} =&  \matr {\tilde L} \matr P_{\phi}^\dagger\matr \Sigma \matr {\tilde R}. \label{ar1}
\end{align}
For notational compactness, let us define $\matr Z_{\beta}\triangleq\matr P_{\phi}\matr {\tilde L}^\dagger\matr \Theta_\beta \matr {\tilde L}\matr P_{\phi}^\dagger$ and $\matr A\triangleq \matr \Sigma \matr {\tilde R}$. Thereby, we can write $\matr {\tilde H}^\dagger\matr \Theta_\beta \matr {\tilde H}= \matr A^\dagger \matr Z_{\beta}\matr A$. Note that $\matr A^\dagger \matr Z_\beta \matr A$  and $\matr Z_\beta \matr A\matr A^\dagger$ have the same eigenvalues. Thus, we have 
\begin{equation}
\det \matr {\tilde H}^\dagger\matr \Theta_\beta \matr {\tilde H} =  \det\matr\Sigma^2\det \matr Z_{\beta}. \label{KEY}
\end{equation}
{\new We complete the derivation of \eqref{general} by plugging \eqref{KEY} in \eqref{ratefg}:}
\begin{align}
\chi_{\matr H}^T(R,\beta R)&=\frac{1}{T}\log_2\frac{\det \matr Z_{1}}{\det\matr Z_\beta}\label{finalres}.
\end{align}
Note also that $\matr Z_{1}={\bf I}$ for $\matr \Lambda_{\rm R}= {\bf I}$. This completes the proof of Theorem~\ref{Theorem2}.

%%%%%%%%%%%%%%%%%%%%%%%%%
\section{Proof of Corollary~\ref{Theorem3}}\label{PTheorem3}
We first show that provided $\matr H^\dagger \matr H$ is unitarily invariant, when $\matr H^\dagger \matr H$ has full rank almost surely, so does $\matr {H}^\dagger\matr P_{\beta}^\dagger\matr P_{\beta} \matr {H}$ too for $\phi\leq \beta$: From \eqref{KEY}
we have 
\begin{equation}
\det \matr {H}^\dagger\matr P_{\beta}^\dagger\matr P_{\beta} \matr {H} = \det\matr \Sigma ^2\det \matr P_{\phi}\matr L^\dagger \matr P_{\beta}^\dagger \matr P_{\beta} \matr L \matr P_{\phi}^\dagger\label{detH}
\end{equation}
where $\matr \Sigma$ is a $T\times T$ diagonal matrix whose diagonal entries are the positive singular values of $\matr H$. By the unitary invariance assumption, $\matr P_{\phi}\matr L^\dagger \matr P_{\beta}^\dagger \matr P_{\beta} \matr L \matr P_{\phi}^\dagger$ is a Jacobi matrix ensemble with a positive determinant for $\phi\leq \beta$\cite{Alain}. Thereby, \eqref{detH} is positive.   

Given $x\sim {\mathcal Be}(a,b)$ we have ${\rm E}[\ln x]=\psi(a)-\psi(a+b)$ where $\psi (\cdot)$ denotes the digamma function. For natural arguments, the digamma function can be expressed as 
\begin{equation}
\psi(n)= \psi (1)+\sum_{l=1}^{n-1}\frac{1}{l}.
\end{equation}
Hence, from \eqref{Betacar} we can write the ergodic rate loss as
\begin{align}
{\rm E}[\chi_{\matr H}^T(R,\beta R)]&= -\frac{1}{T\ln 2}\sum_{t=1}^{T}{\rm E}[\ln \rho_t]\\
&=\frac{1}{T\ln 2}\sum_{t=1}^{T}[\psi(R+1-t)-\psi(\beta R+1-t)]\\
&=\frac{1}{T\ln 2}\sum_{t=1}^{T}\left[\sum_{r=1}^{R-t}\frac{1}{r}-\sum_{r=1}^{\beta R-t} \frac{1}{r}\right]  \label{mstep}\\
&=\frac{1}{T\ln 2}\sum_{t=1}^{T}\sum_{r=\beta R-t+1}^{R-t}\frac{1}{r}.
\end{align}
This completes the derivation of \eqref{ergodic}. 

As regards to derivation of \eqref{Loss}, we first note the almost sure convergence of the limit \cite[Theorem 3.6 and Eq. (4.23)]{Alain}
\begin{equation}
\lim_{T\to \infty}\frac{1}{T}\log_2\det\matr P_{\phi}\matr L^\dagger \matr P_{\beta}^\dagger \matr P_{\beta} \matr L \matr P_{\phi}^\dagger =\int \log_2 (x) \; {\rm dF}_{\matr P_{\beta} \matr U \matr P_{\phi}^\dagger}(x).
\end{equation}
Using \eqref{jocobimultip} we express this limit in terms of binary entropy function:
\begin{equation}
\int \log_2 (x) \;{\rm dF}_{\matr P_{\beta} \matr U \matr P_{\phi}^\dagger}(x)= -\frac{1}{\phi}H(\phi)+\frac{\beta}{\phi}H\left(\frac{\phi}{\beta} \right).
\end{equation}
This completes the derivation of \eqref{Loss}.

%%%%%%%%%%%%%%%%%%%%%%%%%
\section{Proof of Remark~\ref{symmetry}}\label{Psymmetry}
It is sufficient to prove the result for $(R-T)<T$. For the sake of notational compactness, we define $h_p\triangleq\sum_{l=1}^p \frac{1}{l}$ and $g(R,T)\triangleq\ln(2)T\chi^T(R,T)$. Then, from \eqref{mstep} we write
\begin{align}
g(R,T)=&\sum_{t=1}^{T}h_{R-t}-\sum_{t=1}^{T}h_{T-t}\\
=&\sum_{t=1}^{R-T}h_{R-t}+\sum_{t=(R-T)+1}^{T}h_{R-t}\nonumber \\
&-\sum_{t=1}^{2T-R}h_{T-t}-\sum_{t=2T-R+1}^{T}h_{T-t}.
\end{align}
Notice that
\begin{align}
\sum_{t=(R-T)+1}^{T}h_{R-t}&=\sum_{t=1}^{2T-R}h_{T-t} \\
\sum_{t=2T-R+1}^{T}h_{T-t}&=\sum_{t=1}^{R-T}h_{(R-T)-t}.
\end{align}
Thereby, we get
\begin{align}
g(R,T)&=\sum_{t=1}^{R-T}h_{R-t}-\sum_{t=1}^{R-T}h_{(R-T)-t}\\
&=g(R,R-T).
\end{align}
This completes the proof.

%%%%%%%%%%%%%%%%%%%%%%%%%%%%
\section{Proof of Corollary~\ref{Cor3}}\label{PCor3}
%Consider the spectral decomposition of the Gaussian random matrix $\matr X$
%\begin{equation}
%\matr X= \matr U\matr P_{\phi}^\dagger \matr \Sigma \matr V.
%\end{equation}
%The matrices $\matr U\in \CC^{R\times R}$ and $\matr V\in\CC ^{T\times T}$ are Haar distributed since $\matr X$ is bi-unitarily invariant, see \cite{tulino}. Moreover, $\matr \Sigma$ is a $T\times T$ diagonal matrix whose diagonal entries are the positive singular values of $\matr X$. Then, 
Following the same line of argumentation as used to obtain \eqref{finalres} we get
\begin{align}
\log_2\frac{\det \matr {X}^\dagger \Theta_1 \matr {X}}{\det\matr {X}^\dagger \Theta_\beta \matr {X}}=\log_2\frac{\det \matr P_{\phi}\matr {U}^\dagger \Theta_1 \matr {U} \matr P_{\phi}^\dagger}{\det \matr P_{\phi}\matr {U}^\dagger \Theta_\beta \matr {U} \matr P_{\phi}^\dagger}\label{generalG}
\end{align}
where $\matr U$ is a $R\times R$ unitary matrix whose columns are the left singular vectors of the Gaussian random matrix $\matr X$. Since $\matr X^\dagger\matr X$ is unitarily invariant $\matr U$ is Haar distributed. The matrix of the left singular vectors of $\matr {\tilde H}$, i.e. $\matr {\tilde L}$, is Haar distributed too as $\matr {\tilde H}^\dagger \matr {\tilde H}$ is unitarily invariant. Thereby, from \eqref{general} and \eqref{generalG} we have
\begin{equation}
\chi^{T}_{\matr H}(R, \beta R)\sim \frac{1}{T}\log_2\frac{\det \matr {X}^\dagger \Theta_1 \matr {X}}{\det\matr {X}^\dagger \Theta_\beta \matr {X}}.
\end{equation}
Note that the rank of $\matr \Theta_{\beta}$ is $\beta R$. Thus, we can consider the eigenvalue decomposition 
\begin{equation}
\matr \Theta_{\beta}=\matr U_{\beta}^\dagger \matr P_{\beta}^\dagger\matr D_{\beta}\matr P_{\beta}\matr U_{\beta}
\end{equation}
where $\matr U_{\beta}$ is a $R\times R$ unitary matrix. Since $\matr X\sim \matr U_{\beta}\matr X$, we have
\begin{equation}
\matr {X}^\dagger \Theta_\beta \matr {X}\sim \matr {X}^\dagger\matr P_{\beta}^\dagger\matr D_{\beta}\matr P_{\beta}\matr {X}.
\end{equation}
Thereby, we have
\begin{equation}
{\rm E}[\chi^{T}_{\matr H}(R, \beta R)]=\frac{1}{T}{\rm E}\left[\log_2 \frac{\det \matr {X}^\dagger\matr D_{1}\matr {X}}{\det\matr {X}^\dagger\matr P_{\beta}^\dagger\matr D_{\beta}\matr P_{\beta}\matr {X}}\right]
\end{equation}
which completes the proof.

%%%%%%%%%%%%%%%%%%%%%%%%%
\section{Solution of Example \ref{correction}}\label{Examp1}
Note that we do not assume that $\matr H$ has Gaussian entries. However it is well known that for any distribution of the entries of $\matr H$, the distribution function ${\rm F}_{\matr P_{\beta}\matr H}^N$ converges weakly and almost surely to the Mar\u cenko-Pastur law. In other words, we get the same asymtotic results regardless of whether we restrict the entries of $\matr H$ to Gaussian or not. Thus, without loss of generality we can assume that the entries of $\matr H$ are Gaussian, so that $\matr H\matr H^\dagger$ is unitarily invariant. Doing so we have ${\rm S}_{\matr H}(z)=(1+z)^{-1}$\cite{tulino}. Then, we immediately obtain \eqref{iid} from \eqref{jocobimultip}. 
%%%%%%%%%%%%%%%%%%%%%%%%%

\section{Proof of Lemma~\ref{deviation}}\label{Lemma4}
We have ${\alpha}_{\matr P_{\beta}\matr H}=\beta$. Thus 
\begin{equation}
\mathcal{I}_{0}(\gamma;{\rm F}_{\matr P_{\beta}\matr H})=\beta\mathcal{I}_{0}(\gamma;{\rm \tilde F}_{\matr P_{\beta}\matr H})=\beta\mathcal{I}_{0}(\gamma;{\rm F}_{\matr H^\dagger\matr P_{\beta}^\dagger}).\label{ref} 
\end{equation}
Furthermore, with Lemma~\ref{projection} we have
\begin{equation}
{\rm S}_{\matr H^\dagger\matr P_\beta^\dagger}(z)={\rm S}_{\matr H^\dagger}(\beta z)={\rm S}_{\matr H}(\beta z). \label{nec}
\end{equation}
In the sequel we first show that 
\begin{equation}
\int\limits_{0}^{1}\left\vert \log_2\tilde {\rm S}_{\matr P_{\beta}\matr H}(-z)\right\vert\; {\rm d}z=\int\limits_{0}^{1}\left\vert \log_2{\rm S}_{\matr H}(-\beta z)\right\vert \;{\rm d}z<\infty \label{Sleb}
\end{equation}
where $\tilde {\rm S}_{\matr P_{\beta}\matr H}$ is the S-transform of $\tilde {\rm F}_{\matr P_{\beta}\matr H}$. To do so, it is sufficient to show that $\int\limits_{0}^{1}\left\vert \log_2{\rm S}_{\matr H}(-z)\right\vert\; {\rm d}z<\infty$. Since ${\rm F}_{\matr H}$ has a compact support, $\mathcal{I}(\gamma;{\rm F}_{\matr H})$ is finite. Now we show that $\log{x}$ is absolutely integrable over ${\rm F}_{\matr H}$ if, and only if,  $\mathcal{I}(1;{\rm F}_{\matr H})$ and $\Delta\mathcal{I}(1;{\rm F}_{\matr H})$ are finite \cite{hager}:
\begin{align}
\int\limits_{0}^{\infty} \left\vert\log_2(x)\right \vert\; {\rm d}{\rm  F}_{\matr H}(x) =\int\limits_{0}^{1}\log_2\left( \frac{1}{x} \right)\;{\rm d}{\rm F}_{\matr H}(x)\nonumber \\+\int\limits_{1}^{\infty}\log_2(x) \;{\rm d}{\rm F}_{\matr H}(x). \label{L9} 
\end{align}
Thus, we have 
\begin{align}
\int\limits_{0}^{1}\log_2\left(\frac{1}{x} \right)\;{\rm d}{\rm F}_{\matr H}(x)<\infty& \iff\Delta\mathcal I(1;{\rm F}_{\matr H})<\infty, \label{10a} \\
\int\limits_{1}^{\infty}\log_2(x)\; {\rm d}{\rm F}_{\matr H}(x)<\infty &\iff \mathcal I(1;{\rm F}_{\matr H})<\infty.\label{10b}
\end{align}
with $\iff$ implying `'if, and only if". Hence \eqref{L9} is finite. Due to Theorem~\ref{moller} this implies that $\int\limits_{0}^{1}\left\vert \log_2{\rm S}_{\matr H}(-z)\right\vert \;{\rm d}z$ is finite too.  

By invoking Theorem~\ref{moller}, (\ref{ref}) and (\ref{nec}) we obtain
\begin{eqnarray}
\mathcal{I}_{0}(\gamma;{\rm F}_{\matr P_{\beta}\matr H})
=\beta\log_2\gamma-\beta\int _{0}^{1}\log_2{\rm S}_{\matr H}(-\beta z)\;{\rm d}z. \label{harmon1}
\end{eqnarray}
Due to \eqref{Sleb}, it follows from the linearity property of the Lebesgue integral that
\begin{align}
\Delta \mathcal L(\beta;{\rm F}_{\matr H})&= \mathcal I_{0}(\gamma;{\rm F}_{\matr P_\beta\matr H})-\beta\mathcal I_{0}(\gamma;{\rm F}_{\matr H})\\  &= -\beta\int _{0}^{1}\log_2\frac{{\rm S}_{\matr H}(-\beta z)}{{\rm S}_{\matr H}(-z)}\;{\rm d}z. 
\end{align}
This completes the proof.

%%%%%%%%%%%%%%%%%%%%%%%%%
\section{Proof of Remark~\ref{defh}} 
Invoking Lemma~\ref{projection} we can write
\begin{eqnarray} 
 {\rm S}_{\matr H}(-t)={\rm S}_{\matr H^\dagger}(-t)=\lim_{z\to -1^{+}}{{\rm S}_{\matr H^\dagger\matr P_t^\dagger}(z)}.
\end{eqnarray}
Since $\matr H$ has almost surely full rank, $\alpha_{\matr H\matr P_{t}^\dagger}=1$, so that ${\rm \tilde F}_{\matr P_t\matr H}={\rm F}_{\matr H^\dagger\matr P_t^\dagger}$. Then from Lemma~\ref{Smap} we have
\begin{equation}
 \lim_{z\to -1^{+}}{\rm S}_{\matr H^\dagger \matr P_{t}^\dagger}(z)= \int \frac{1}{x}\; {\rm dF}_{\matr H^\dagger \matr P_{t}^\dagger}(x).\label{Sharmonic}
\end{equation}
This completes the proof. 

%%%%%%%%%%%%%%%%%%%%%%%%%
\section{Proof of Remark~\ref{devsup}}\label{devsupproof}
For the sake of notational simplicity we introduce
\begin{align}
\matr Y_{\beta}&\triangleq \matr {\bf I}+\gamma \matr P_{\beta}\matr H\matr H^\dagger \matr P_{\beta}^\dagger \\
&=\matr P_{\beta}(\matr {\bf I}+\gamma \matr H\matr H^\dagger)\matr P_{\beta}^\dagger \\
&=\matr P_{\beta}\matr Y_1\matr P_{\beta}^\dagger.
\end{align}
It follows that $\matr Y_1$ is unitarily invariant since $\matr H\matr H^\dagger$ is. Furthermore, since  $\matr H \matr H^\dagger$ has a compactly support LED so does $\matr Y_1$. Thus $\matr Y_1$ is asymptotically free of $\matr P_\beta^\dagger \matr P_\beta$ \cite{hiai}. Then, with Lemma~\ref{projection} we have in the limit $N\to \infty$
\begin{equation}
{\rm S}_{\sqrt{\matr Y_\beta}}(z)={\rm S}_{\sqrt{\matr Y_1}}(\beta z).\label{Like}
\end{equation}
Here we note that ${\rm S}_{\sqrt{\matr Y_\beta}}(z)$ is strictly decreasing on $(-1,0)$ if, and only if, ${\rm F}_{\sqrt{\matr Y_\beta}}$ is not a Dirac distribution function, see Lemma~\ref{Smap}. 

We recall the following property of $\eta_{\matr H}(\gamma)$ see \eqref{eta}\cite{tulino}:
\begin{equation}
\frac{{\rm d}\{\mathcal I(\gamma;{\rm F}_{\matr P_\beta\matr H})-\beta\mathcal I(\gamma;{\rm F}_{\matr H})\}}{{\rm d}\gamma}=\frac{1-\eta_{\matr P_\beta \matr H}- \beta (1-\eta_{\matr H})}{\gamma \ln 2},
\end{equation}
where for convenience $\eta_{\matr H}$ is short for $\eta_{\matr H}(\gamma)$. Hence, in order to prove the remark it is sufficient to show that 
\begin{equation}
(1-\beta)+\beta\eta_{\matr H}-\eta_{\matr P_\beta \matr H}\geq 0\label{ineq1}
\end{equation} 
where the equality holds when $\beta=1$. Furthermore, by using \cite[Lemma 2.26]{tulino} we have
\begin{equation}
\eta_{\matr P_\beta \matr H}=(1-\beta)+\beta\eta_{\matr H^\dagger \matr P_\beta^\dagger}.
\end{equation} 
Thus, the right-hand side of (\ref{ineq1}) is equal to $\beta(\eta_{\matr H}-\eta_{\matr H^\dagger\matr P_\beta^\dagger})$. Therefore we are left with proving $\eta_{\matr H}\geq \eta_{\matr H^\dagger\matr P_\beta^\dagger}$. Firstly, remark that 
\begin{equation}
\eta_{\matr H^\dagger\matr P_\beta^\dagger}=\int \frac{1}{x}\;{\rm dF}_{\sqrt{\matr Y_\beta}}(x).
\end{equation} 
Then, by using (\ref{Sharmonic}) and \eqref{Like} we obtain  
\begin{align}
\eta_{\matr H^\dagger\matr P_\beta^\dagger}&=\lim_{z\to -1^{+}}{\rm S}_{\sqrt{\matr Y_\beta}}(z)\\ 
&=\lim_{z\to -1^{+}}{\rm S}_{\sqrt{\matr Y_1}}(\beta z)\\
&={\rm S}_{\sqrt{\matr Y_1}}(-\beta ) \;, \quad 0<\beta<1
\end{align}
which is strictly increasing with $\beta$, see Lemma~\ref{Smap}. This completes the proof.
\section{Proof of Theorem~3}\label{The3}
The matrices $\matr X\matr X^\dagger$, $\matr Y\matr Y^\dagger$ and $ \matr P_\beta^\dagger\matr P_\beta$ are asymptotically free \cite{hiai}.  Then, from Lemma \ref{deviation} and the linearity property of the Lebesgue integral we have
\begin{align}
\Delta \mathcal L(\beta;{\rm F}_{\matr X\matr Y})&=-\beta\int\limits_{0}^{1}\log_2\frac{{\rm S}_{\matr X}(-\beta z){\rm S}_{\matr Y}(-\beta z)}{{\rm S}_{\matr X}(-z){\rm S}_{\matr Y}(-z)}\;{\rm d}z\\
&=\Delta \mathcal L(\beta;{\rm F}_{\matr X})+\Delta\mathcal L(\beta;{\rm F}_{\matr Y}).\label{final}
\end{align}
%%%%%%%%%%%%%%%%%%%%%%%%%

\section{Proof of Remark~\ref{CTadd}}\label{R7}
For an $N\times N$ matrix $\matr A$, we define 
\begin{equation}
\phi (\matr A)\triangleq \lim_{N\to \infty}\frac{1}{N}{\rm tr}(\matr A)
\end{equation}
whenever the limit exists. Since $\matr X\matr X^\dagger$ and $\matr Y\matr Y^\dagger$ are asymptotically free, we have (see \cite[Eq. (120)]{ralfc})
\begin{align}
\phi (\matr X^\dagger\matr Y^\dagger\matr Y\matr X)&=\phi (\matr X^\dagger\matr X)\phi (\matr Y^\dagger\matr Y)\label{iden1}\\
\phi ((\matr X^\dagger \matr Y^\dagger \matr Y \matr X)^2)&= \phi(\matr X^\dagger \matr X)^2 \phi ((\matr Y^\dagger \matr Y)^2)\nonumber \\ &+\phi(\matr Y^\dagger \matr Y)^2 \phi ((\matr X^\dagger \matr X)^2)\nonumber \\&-\phi(\matr X^\dagger \matr X)^2\phi(\matr Y^\dagger \matr Y)^2 \label{iden2}.
\end{align}
Furthermore, from \cite[Theorem~2.1]{Burthesis} for $\matr A\in \{\matr X, \matr Y, \matr X \matr Y\}$ we have almost surely
\begin{equation}
\lim_{N\to \infty} {(\matr A^\dagger \matr A)_{ii}}\to \phi (\matr A^\dagger \matr A),~ \forall i.  \label{result3}
\end{equation}
Inserting \eqref{result3} in the definition of the crosstalk ratio in \eqref{CTR}, we get
for $\matr A\in \{\matr X, \matr Y, \matr X \matr Y\}$ that
\begin{equation}
{\rm CT}_{\matr A}= \frac{\phi ((\matr A^\dagger \matr A)^2)}{2\phi(\matr A^\dagger \matr A)^2}-\frac{1}{2}.\label{ACTR}
\end{equation}
We complete the proof by plugging \eqref{iden1} and \eqref{iden2} in \eqref{ACTR} for $\matr A=\matr {XY}$:
\begin{align}
{\rm CT}_{\matr X\matr Y}&=\frac{\phi(\matr X^\dagger \matr X)^2 \phi ((\matr Y^\dagger \matr Y)^2)+\phi(\matr Y^\dagger \matr Y)^2 \phi ((\matr X^\dagger \matr X)^2)}{2\phi (\matr X^\dagger\matr X)^2\phi (\matr Y^\dagger\matr Y)^2}-1\\
&= {\rm CT}_{\matr X}+{\rm CT}_{\matr Y}.
\end{align}

\bibliographystyle{IEEEtran}
\bibliography{liter}
\vspace*{-2\baselineskip}

\begin{IEEEbiographynophoto}{Burak~\c{C}akmak} was born in Istanbul, Turkey, 1986. He received the B.Eng. degree from Uluda\u{g} University, Turkey in 2009, M.Sc. degree from Norwegian University of Science and Technology, Norway in 2012 and Ph.D degree from Aalborg University, Denmark in 2017.  Dr. \c{C}akmak  is a postdoctoral researcher at the Department of Computer Science, Technical University of Berlin, Germany. His research interests include random matrix theory, communication theory, statistical physics of disorder systems, machine learning and Bayesian inference.  
\end{IEEEbiographynophoto}
\vspace*{-2\baselineskip}
\begin{IEEEbiographynophoto}{Ralf~R.~M\"{u}ller} (S'96--M'03--SM'05)
was born in Schwabach, Germany, 1970. He received the Dipl.-Ing. and Dr.-Ing. degree with distinction from Friedrich-Alexander-Universit\"{a}t (FAU) Erlangen-N\"{u}rnberg in 1996 and 1999, respectively. From 2000 to 2004, he directed a research group at The Telecommunications Research Center Vienna in Austria and taught as an adjunct professor at TU Wien. In 2005, he was appointed full professor at the Department of Electronics and Telecommunications at the Norwegian University of Science and Technology in Trondheim, Norway. In 2013, he joined the Institute for Digital Communications at FAU Erlangen-N\"{u}rnberg in Erlangen, Germany. He held visiting appointments at Princeton University, US, Institute Eurécom, France, University of Melbourne, Australia, University
of Oulu, Finland, National University of Singapore, Babes¸-Bolyai University, Cluj-Napoca, Romania, Kyoto University, Japan, FAU Erlangen-N\"{u}rnberg, Germany, and TU M\"{u}nchen, Germany.
Prof. M\"{u}ller received the Leonard G. Abraham Prize (jointly with Sergio Verd\'{u}) for the paper ``Design and analysis of low-complexity interference mitigation on vector channels'' from the IEEE Communications Society. He was presented awards for his dissertation ``Power and bandwidth efficiency of multiuser systems with random spreading'' by the Vodafone Foundation for Mobile Communications and the German Information Technology Society (ITG). Moreover, he received the ITG award for the paper ``A random matrix model for communication via antenna arrays'' as well as the Philipp-Reis Award (jointly with Robert Fischer). Prof. M\"{u}ller served as an associate editor for the IEEE TRANSACTIONS ON INFORMATION THEORY from 2003 to 2006 and as an executive editor for the IEEE TRANSACTIONS ON WIRELESS COMMUNICATIONS from 2014 to 2016. 
\end{IEEEbiographynophoto}

\vspace*{-2\baselineskip}
\begin{IEEEbiographynophoto}{Bernard H. Fleury} (M'97--SM'99) received the Diplomas in Electrical Engineering and in Mathematics in 1978 and 1990 respectively and the Ph.D. Degree in Electrical Engineering in 1990 from the Swiss Federal Institute of Technology Zurich (ETHZ), Switzerland. Since 1997, he has been with the Department of Electronic Systems, Aalborg University, Denmark, as a Professor of Communication Theory. From 2000 till 2014 he was Head of Section, first of the Digital Signal Processing Section and later of the Navigation and Communications Section. From 2006 to 2009, he was partly affiliated as a Key Researcher with the Telecommunications Research Center Vienna (ftw.), Austria. During 1978–1985 and 1992–1996, he was a Teaching Assistant and a Senior Research Associate, respectively, with the Communication Technology Laboratory, ETHZ. Between 1988 and 1992, he was a Research Assistant with the Statistical Seminar at ETHZ. Prof. Fleury’s research interests cover numerous aspects within communication theory, signal processing, and machine learning, mainly for wireless communication systems and networks. His current scientific activities include stochastic modeling and estimation of the radio channel, especially for large systems (operating in large bandwidths, equipped with large antenna arrays, etc.) deployed in harsh conditions (e.g. in highly time-varying environments); iterative message-passing processing (with focus on the design of efficient feasible architectures for wireless receivers); localization techniques in wireless terrestrial systems; and radar signal processing. Prof. Fleury has authored and coauthored more than 150 publications and is co-inventor of 6 filed or published patents in these areas. He has developed, with his staff, a high-resolution method for the estimation of radio channel parameters that has found a wide application and has inspired similar estimation techniques both in academia and in industry. 
\end{IEEEbiographynophoto}
\end{document}